\documentclass[runningheads]{llncs}

\usepackage[utf8]{inputenc} 
\usepackage[T1]{fontenc}    
\usepackage{hyperref}       
\usepackage{url}            
\usepackage{booktabs}       
\usepackage{amsfonts}       
\usepackage{nicefrac}       
\usepackage{microtype}      
\usepackage{xcolor}         
\usepackage{graphicx}
\usepackage{algorithm}
\usepackage{algorithmic}
\usepackage{amsmath}
\usepackage{amssymb}
\let\oldendproof\endproof
\renewcommand{\endproof}{\qed\oldendproof}

\usepackage{tikz}
\usetikzlibrary{positioning}
\usetikzlibrary{calc}
\usetikzlibrary{shapes.geometric}
\usetikzlibrary{positioning, shapes.geometric, backgrounds}
\usepackage{subcaption}

\spnewtheorem{observation}{Observation}{\bfseries}{\itshape}
\makeatletter
\def\@spopargbegintheorem#1#2#3#4#5{\trivlist
  \item[\hskip\labelsep{#4#1\ #2}]{\normalfont(#3)\@thmcounterend\ }#5}
\makeatother

\newcommand{\fundingnote}[1]{%
  \begingroup
  \renewcommand{\thefootnote}{}%
  \footnotetext{#1}%
  \addtocounter{footnote}{-1}%
  \endgroup
}

\title{Strictly Unfriendly $k$-Partitions: Sharp Degree Thresholds and ETH-Based Lower Bounds}

\titlerunning{Strictly Unfriendly $k$-Partitions}

\author{Sanjay Jain\orcidID{0000-0001-6798-8330}
\and
Frank Stephan\orcidID{0000-0001-9152-1706}
\and
Haoyun Tang\orcidID{0009-0001-5410-0221}}

\institute{School of Computing, National University of Singapore, Singapore\\
\email{sanjayjain@nus.edu.sg, fstephan@nus.edu.sg, e1154532@u.nus.edu}}

\begin{document}

\maketitle

\fundingnote{Sanjay Jain and Frank Stephan were supported by Singapore Ministry of Education (MOE) AcRF Tier 2 grant MOE-000538-01. Additionally, Sanjay Jain was supported by NUS grant E-252-00-0021-09.}

\begin{abstract}
We present a complete complexity classification and fine-grained analysis for the Strictly Unfriendly $k$-Partition problem ($\text{SU}k\text{P}$), which asks whether the vertices of a graph can be partitioned into $k$ classes such that every vertex has strictly more neighbors in each of the other $k-1$ classes than in its own. We first establish a sharp tractability-intractability threshold with respect to the maximum degree $\Delta$: for $k \in \{2, 3\}$, $\text{SU}k\text{P}$ is solvable in polynomial time when $\Delta \le 2$, but becomes $\mathbf{NP}$-hard and ETH-hard immediately on subcubic graphs ($\Delta = 3$), resolving the degree limitations in prior work and establishing subcubic graphs as the precise frontier of intractability.

Furthermore, under the Exponential Time Hypothesis (ETH), we establish the first fine-grained lower bounds via direct reductions from $(3,3)$-SAT. On general graphs, we establish a uniform lower bound across all partition parameters $k \ge 2$, revealing a striking \emph{complexity convergence} where the core exponential complexity remains invariant despite technical divergences in gadget constructions. On subcubic graphs, we formally quantify the "cost of sparsity," deriving explicit lower bound constants to demonstrate how enforced structural degree restrictions degrade reduction efficiency.

\end{abstract}


\section{Introduction}

The study of graph partitions with local constraints is a classic theme in structural graph theory and theoretical computer science. A foundational problem in this domain is the \emph{Unfriendly Partition} problem, introduced by Cowen and Emerson~\cite{greenwade95}, which asks whether the vertices of a graph $G=(V,E)$ can be partitioned into $k$ sets such that every vertex has at least as many neighbours in other sets as it has in its own. While an unfriendly 2-partition always exists for every finite graph via a simple potential function argument, the problem has inspired extensive research on infinite graph classes. Notably, Shelah and Milner~\cite{Shelah_Milner_1990} disproved the existence of unfriendly 2-partitions for general infinite graphs, while simultaneously proving that an unfriendly 3-partition always exists. Subsequent studies further expanded existence results to countable infinite degree graphs~\cite{AHARONI19901}, rayless graphs~\cite{bruhn2010every}, and graphs avoiding certain infinite configurations~\cite{berger2017unfriendly}.

For finite graphs, where existence is guaranteed, research has shifted toward algorithmic complexity and constrained variants. Sarnath and He~\cite{SARNATH1990343} presented a polynomial-time algorithm for finding the lexicographically first unfriendly partition, whereas Sarnath~\cite{sarnath1992dtml} established LOGSPACE-hardness under $\text{NC}^1$ reductions. To explore the boundaries of intractability, several NP-complete variants of graph partitioning under local or membership constraints were introduced (e.g., \cite{feder1999complexity}).

In this paper, we focus on the \textbf{Strictly Unfriendly Partition (SUP)} problem, introduced by Haynes et al.~\cite{HAYNES2015155}. A $k$-partition is strictly unfriendly if every vertex has \emph{strictly more} neighbours in \emph{each} of the other $k-1$ partition classes than in its own. The transition from "at least as many" to "strictly more" fundamentally alters the problem's behavior: an SU$k$P is no longer guaranteed to exist for every finite graph, giving rise to the natural decision problem $\text{SU}k\text{P}$. Haynes et al.~\cite{HAYNES2015155} provided initial structural observations on specific graph families (e.g., $k$-cliques admit an SU$k$P, whereas $(k+1)$-cliques do not).

From a game-theoretic perspective, $\text{SU}k\text{P}$ corresponds to stable colourings in anti-coordination games, as studied by Kun et al.~\cite{kun2013anti}. While Kun et al.~\cite{kun2013anti} established the NP-hardness of $\text{SU}k\text{P}$ for general $k$, their constructions leave high-degree instances and leave two key aspects unexplored:
\begin{enumerate}
    \item \emph{Tight Degree Constraints for $k \in \{2, 3\}$:} For $k=2$, Kun et al.'s reduction produces vertices whose degrees scale with variable occurrences ($\Delta \ge 6$ or unbounded); for $k \ge 3$, their reduction incorporates clique gadgets ($K_{k-2}$) per edge, causing the maximum degree $\Delta$ to scale linearly with $k$ (e.g., $\Delta \ge 8$ for $k=3$). Consequently, the exact degree-based tractability threshold ($\Delta = 3$) remained open for $k \in \{2, 3\}$.
    \item \emph{Fine-Grained Quantitative Bounds:} Existing proofs focused solely on NP-hardness, leaving ETH-based quantitative lower bounds and the reduction efficiency loss under degree restrictions entirely uncharted.
\end{enumerate}

\medskip
\noindent \textbf{Our Contributions.}
We address these gaps by providing a complete complexity classification and fine-grained lower bound analysis for $\text{SU}k\text{P}$. Assuming ETH, let $c_{A}$ denote the lower bound constant of a problem $A$ (as defined in Definition~\ref{def_lbc}), where $c_{(3,3)\text{-SAT}}$ corresponds to $(3,3)$-SAT. Our main results are as follows:

\begin{itemize}
    \item \textbf{A Sharp Degree-Based Dichotomy ($\Delta \le 2$ vs. $\Delta \geq 3$):} 
    We completely pin down the tractability threshold of $\text{SU}k\text{P}$ with respect to the maximum degree $\Delta$ for $k \in \{2, 3\}$. We observe that $\Delta \ge k-1$ is a necessary condition for a graph to admit an $\text{SU}k\text{P}$, rendering the problem trivial when this bound is violated. Crucially, for $k \in \{2, 3\}$, we prove a sharp complexity phase transition at $\Delta = 3$: while $\text{SU}k\text{P}$ is solvable in polynomial time ($\mathbf{P}$) when $\Delta \le 2$, it becomes $\mathbf{NP}$-hard immediately on subcubic graphs ($\Delta = 3$). This resolves the unconstrained degree bottlenecks present in prior work, establishing subcubic graphs as the precise frontier of intractability for these pivotal cases.

    \item \textbf{Uniform Lower Bounds and Complexity Convergence:} 
    By constructing direct, efficiency-preserving reductions from $(3,3)$-SAT, we establish a \textbf{uniform lower bound} of $c_{\text{SU}k\text{P}} \ge c_{(3,3)\text{-SAT}}^{2/7}$ for all $k \ge 2$ on general graphs. This demonstrates that the core exponential complexity of the strictly unfriendly condition remains invariant with respect to the partition parameter $k$ when degree constraints are unconstrained.

    \item \textbf{Subcubic Lower Bounds and the "Cost of Sparsity":} 
    For degree-restricted instances, we present tailored gadget constructions on subcubic graphs ($\Delta=3$) for the pivotal cases $k=2$ and $k=3$, establishing ETH-based lower bound constants of $c_{\text{SU2P}} \ge c_{(3,3)\text{-SAT}}^{1/13}$ and $c_{\text{SU3P}} \ge c_{(3,3)\text{-SAT}}^{1/38}$. These results formally quantify the trade-off between graph sparsity and reduction efficiency.
\end{itemize}

The remainder of this paper is organized as follows. Section~\ref{sec_pre} formalizes the notation and complexity-theoretic foundations. Sections~\ref{sec_su2p} and~\ref{sec_su3p} present the core subcubic reductions for SU2P and SU3P. Section~\ref{sec_general_hardness} generalizes these findings to any fixed $k \geq 2$ and establishes the uniform lower bound on general graphs. Section~\ref{sec_con} concludes the paper. Detailed proofs are deferred to Appendix (Section~\ref{sec_app}).

\section{Preliminaries}
\label{sec_pre}
In this section, we formalize notation, define the Strictly Unfriendly Partition problem, and introduce the complexity-theoretic foundations under the Exponential Time Hypothesis (ETH).

\subsection{Problem Definition}
Let $G=(V, E)$ be a simple, undirected finite graph. For any $v \in V$, let $N(v) = \{u \mid \{u, v\} \in E\}$ denote its neighbourhood and $d(v) = |N(v)|$ its degree. Let $L = \{0, 1, \dots, k-1\}$ be the set of labels. A $k$-partition of $V$ is represented by a labeling function $\pi: V \to L$. For any label $i \in L$, let $N_i(v) = \{u \in N(v) \mid \pi(u) = i\}$ and $d_i(v) = |N_i(v)|$.

\begin{definition}[Strictly Unfriendly $k$-Partition, SU$k$P]
    A $k$-partition $\pi$ is \emph{strictly unfriendly} if for every vertex $v \in V$ and label $i \in L \setminus \{\pi(v)\}$, it holds that $d_{\pi(v)}(v) < d_i(v)$.
\end{definition}
A vertex $v$ is \emph{legal} under $\pi$ if the above condition holds, and $G$ is \emph{legally $k$-partitionable} if it admits an SU$k$P.

\noindent \textbf{Remark.}
It is immediate from the definition that if $G$ admits an $\text{SU}k\text{P}$, every vertex must have at least one neighbour in each of the other $k-1$ partition classes. Thus, $\Delta(G) \ge k-1$ is an immediate necessary condition, rendering instances with $\Delta < k-1$ trivially NO-instances.

\subsection{Complexity Foundations}

\noindent \textbf{Boolean Satisfiability (SAT).}
We denote by $(s, k)$-SAT the class of CNF formulas where each clause contains at most $k$ literals and each variable appears in at most $s$ clauses. To streamline our reductions, all $(3,3)$-SAT instances are preprocessed, ensuring that every clause contains at least 2 literals and every variable appears 2 or 3 times in total, with both positive and negative occurrences. This establishes a tight linear relation between clauses $m$ and variables $n$:

\begin{lemma} \label{lem_nm}
    In a preprocessed $(3, 3)$-SAT instance with $n$ variables and $m$ clauses, $\frac{2}{3}n \leq m \leq \frac{3}{2}n$.
\end{lemma}
\begin{proof}
    Let $\ell$ denote total literal occurrences. Preprocessing guarantees $2n \le \ell \le 3n$ and $2m \le \ell \le 3m$, yielding $\frac{2}{3}n \le m \le \frac{3}{2}n$.
\end{proof}

\medskip
\noindent \textbf{Exponential Time Hypothesis (ETH).}
\begin{definition}[Lower Bound Constant] \label{def_lbc}
    For a decision problem $\text{A}$ with a primary complexity parameter $n$ and an auxiliary parameter $m$, we define its \emph{lower bound constant} $c_{\text{A}}$ as:
    $$c_{\text{A}} = \inf \{d \mid \text{A is solvable in } O(d^n + \text{poly}(n,m)) \text{ time}\}.$$
\end{definition}

Here, $n$ denotes the main parameter governing the exponential time bound (e.g., the number of variables in a SAT formula or vertices in a graph), whereas $m$ represents an auxiliary parameter (e.g., clauses or edges). The Exponential Time Hypothesis (ETH)~\cite{impagliazzo2001complexity} conjectures $c_{3\text{-SAT}} > 1$. Since $m = O(n)$ in $(3,3)$-SAT (Lemma~\ref{lem_nm}), $c_{(3,3)\text{-SAT}}$ is parameterized solely by $n$. All fine-grained lower bounds in this paper are expressed directly in terms of $c_{(3,3)\text{-SAT}}$. While ETH is originally stated for $3$-SAT, $(3,3)$-SAT remains ETH-hard with a linearly preserved lower bound constant (Detailed proofs are deferred to Appendix, which is Section~\ref{sec_app_prf2}):

\begin{lemma}[Folklore] \label{lem_33}
    $(3, 3)$-SAT is ETH-hard; specifically, $c_{3\text{-SAT}}^{1/(d+1)} \leq c_{(3, 3)\text{-SAT}}$ for some constant $d \in \mathbb{N}^+$.
\end{lemma}

\section{Hardness of Strictly 2-Partition (SU2P)}
\label{sec_su2p}

While SU2P is straightforwardly solvable in polynomial time for graphs with maximum degree two (where it requires each vertex to have its entire neighbourhood in the complement partition class), the problem exhibits a sharp transition to intractability as the degree increases. In this section, we establish the hardness of SU2P on sub-cubic graphs, thereby pinpointing the tractability threshold with respect to the maximum degree.
\begin{theorem} \label{thm_su2p}
SU2P is NP-hard on graphs with maximum degree three. Assuming ETH, it cannot be solved in $2^{o(n)}$ time; specifically, $c_{(3,3)\text{-SAT}}^{1/13} \leq c_{\text{SU2P}}$.
\end{theorem}
The proof is established via a linear reduction from $(3,3)$-SAT. Since the number of vertices in our constructed graph is linear in the number of variables of the $(3,3)$-SAT instance, the ETH-hardness follows directly from Lemma~\ref{lem_33}.

\medskip
\noindent \textbf{Key Observations for SU2P.}
To simplify the reduction and the subsequent proof of correctness, we formalize the local constraints imposed by the strictly unfriendly 2-partition condition on low-degree vertices. Throughout this section, the label set is $L=\{0,1\}$, where if a vertex $u$ has label $a \in L$, then $b = 1-a$ denotes the complementary label. The following observations are straightforward; their formal proofs are deferred to appendix (Section~\ref{sec_app_prf2}).

\begin{observation}[Constraints on Degree-2 Vertices] \label{ob_SU2P_d2}
In any legal SU2P, a degree-2 vertex must have both neighbours to have the label opposite its own.
\end{observation}

\begin{observation}[Constraints on Degree-3 Vertices] \label{ob_SU2P_d3}
In any legal SU2P, a degree-3 vertex $u$ with label $a$ can have at most one neighbour with label $a$.
\end{observation}

\medskip
\noindent \textbf{Reduction Overview.}
Given a $(3,3)$-SAT instance $\phi$, we construct a graph $G_\phi$ consisting of three primary types of gadgets:
\begin{enumerate}
    \item Variable Gadgets: Represent the truth assignment of each variable and ensure consistency between its occurrences with different polarities.
    \item Clause Gadgets: Enforce the satisfiability of each clause.
    \item Constant Gadgets: Provide fixed-label neighbours to other gadgets, acting as truth-value anchors.
\end{enumerate}
In the following figures, \emph{square nodes} represent vertices whose labels are forced to a constant value by the forcing gadgets, while \emph{circular nodes} represent standard vertices whose labels depend on the truth assignment. Proofs of the propositions are deferred to appendix (Section~\ref{sec_app_prf2}).

\medskip
\noindent \textbf{Constant Gadget.} To provide consistent label vertices throughout the reduction, we construct an even cycle $C_{2k}$ consisting of vertices $u_1, u_2, \dots, u_{2k}$ (Figure~\ref{fig_SU2P_d3_cst} in Appendix). In the following content, all vertex subscripts are taken modulo $2k$ with representative indices in $\{1, \dots, 2k\}$, such that $u_0$ is identified with $u_{2k}$. The vertices are categorized as follows:
\begin{itemize}
    \item Internal vertices: Vertices with even indices $u_{2j}$ ($1\le j \le k$) have degree two, connected only to their neighbours $u_{2j-1}$ and $u_{2j+1}$ within the cycle.
    \item Anchor vertices: Vertices with odd indices $u_{2j+1}$ ($1\le j \le k$) have degree three, connected to $u_{2j}, u_{2j+2}$ and an additional external neighbour outside the gadget.
\end{itemize}

This construction allows us to give an arbitrary number of a specific label neighbour to external vertices. Note that for each constant neighbour provided, the gadget introduces exactly two vertices on average (one $u_{2j}$ and one $u_{2j+1}$), which we denote as the amortized cost per forced vertex.
\begin{proposition} \label{prop_SU2P_d3_cst}
    A partition $\pi$ of the constant gadget is legal if and only if the labels alternate around the cycle, i.e., $\pi(u_i) \neq \pi(u_{i+1})$ for all $1 \leq i \leq 2k$.
\end{proposition}
Without loss of generality, we fix the label of all anchor vertices $u_{2j+1}$ to be $1$. We use these vertices as sources of the constant-1 neighbour in the subsequent construction.

\medskip
\noindent \textbf{Variable Gadget.} For each variable $x$ in the $(3,3)$-SAT instance $\phi$, we construct a variable gadget to represent its truth assignment and ensure consistency across its literal occurrences. As established in the preprocessing, each variable appears at most twice in one polarity and once in the other. Without loss of generality, let $x$ be the literal appearing twice (if any) and $\neg x$ be the literal appearing once. The gadget, shown in Figure~\ref{fig_SU2P_d3_var}, consists of four vertices:
\begin{itemize}
    \item Literal vertices: $v_1, v_2$ (representing upto two occurrences of $x$) and $v_3$ (representing the single occurrence of $\neg x$). These are degree-3 vertices connected to their respective clause gadgets.
    \item Auxiliary vertex: $v_4$ is a degree-2 vertex connected to $v_1$ and $v_2$, serving as a structural constraint to synchronize the labels of $v_1$ and $v_2$.
\end{itemize}

\begin{figure}
    \centering
    \begin{subfigure}[b]{0.4\textwidth}
        \centering
    \resizebox{0.99\textwidth}{!}{
    \begin{tikzpicture}[
normal/.style = {draw, regular polygon,regular polygon sides=360, very thick},
fixed/.style={draw, regular polygon,regular polygon sides=4, very thick},
balance/.style = {draw, regular polygon,regular polygon sides=360, fill=black},
]
\node[] (u1) [] {to clause};
\node[normal] (v1) [right=1cm of u1, label=above:$x$] {$v_1$};
\node[normal] (v2) [right=1cm of v1, label=above:$x$] {$v_2$};
\node[normal] (v3) [below=1cm of v1, label=below:$\lnot x$] {$v_3$};
\node[normal] (v4) [below=1cm of v2, label=below:$\lnot x$] {$v_4$};
\node[] (u2) [right=1cm of v2] {to clause};
\node[] (u3) [left=1cm of v3] {to clause};

\draw (v1) -- (u1);
\draw (v1) -- (v3);
\draw (v1) -- (v4);
\draw (u2) -- (v2);
\draw (v2) -- (v3);
\draw (v2) -- (v4);
\draw (u3) -- (v3);

\end{tikzpicture}
}
    \caption{}
    \label{fig_SU2P_d3_var}
    \end{subfigure}
    \hfill 
    \begin{subfigure}[b]{0.45\textwidth}
        \centering
        \resizebox{0.99\textwidth}{!}{
    \begin{tikzpicture}[
normal/.style = {draw, regular polygon,regular polygon sides=360, very thick},
fixed/.style={draw, regular polygon,regular polygon sides=4, very thick, minimum size=1.2cm},
balance/.style = {draw, regular polygon,regular polygon sides=360, fill=black},
]
\node[normal] (v1) [label=above:$0$] {$v_1$};
\node[normal] (v2) [right=0.5cm of v1, label=above:$1$] {$v_2$};
\node[normal] (v3) [right=0.5cm of v2, label=above:$0$] {$v_3$};
\node[fixed] (u2) [right=0.5cm of v3, label=above:$1$] {$u'$};
\node[] (v8) [below right=1cm of v1] {$z$};
\node[normal] (v4) [below left=1cm of v1] {$v_4$};
\node[normal] (v5) [below left=1cm of v4] {$v_5$};
\node[fixed] (u1) [below right=1cm of v4, label=above:$1$] {$u$};
\node[] (v6) [below left=1cm of v5] {$x$};
\node[] (v7) [below right=1cm of v5] {$y$};

\draw (v3) -- (u2);
\draw (v3) -- (v2);
\draw (v1) -- (v2);
\draw (v1) -- (v4);
\draw (v1) -- (v8);
\draw (v4) -- (u1);
\draw (v4) -- (v5);
\draw (v6) -- (v5);
\draw (v7) -- (v5);

\end{tikzpicture}
}
    \caption{Labels $\pi(v)$ are indicated above the vertices.}
    \label{fig_SU2P_d3_3cls}
    \end{subfigure}
    
    \caption{}
\end{figure}

In this construction, the label of a literal vertex $\pi(v_i)$ directly corresponds to the truth value assigned to that literal.
\begin{proposition} \label{prop_SU2P_d3_var}
    A variable gadget is legally partitioned if and only if the literal vertices satisfy the consistency and polarity requirements: $\pi(v_1) = \pi(v_2)$ and $\pi(v_3) = 1 - \pi(v_1)$.
\end{proposition}

\medskip
\noindent \textbf{Clause Gadget.} For each clause $C = (x \lor y \lor z)$, we construct a clause gadget as illustrated in Figure~\ref{fig_SU2P_d3_3cls}. This construction utilizes five internal vertices $v_1, \dots, v_5$ and two constant-1 neighbours provided by anchor vertices $u, u'$ from the constant gadget. The literal vertices $x, y, z$ (from their respective variable gadgets) are connected to $v_1$ and $v_5$ as shown. On average, each 3-literal clause introduces 9 vertices: five internal vertices and four amortized vertices from the constant gadget ($2 \times 2$ vertices per constant-1 neighbour).

\begin{proposition} \label{prop_SU2P_d3_3cls}
    A clause gadget is legally partitioned if and only if at least one of its literal vertices $\{x, y, z\}$ is assigned the label $1$ (True).
\end{proposition}
For clauses containing only two literals, say $(x \lor y)$, the gadget is simplified by replacing the connections of $v_1$ by $x$ and $y$ (shown in Appendix, Figure~\ref{fig_SU2P_d3_2cls}). This gadget uses 5 vertices, and thus $5/2=2.5$ vertices per literal. As a 3-clause gadget uses $9/3=3$ vertices per clause, the worst-case size of the reduction is bounded by the 3-clause case (as the more the vertices are used through the reduction, the less the ETH-based bound is). Consequently, the ETH-based lower bound remains robust under this omission.

\medskip
\noindent \textbf{Full Reduction.} With the gadgets defined, we now summarize the reduction and analyze the resulting complexity lower bound.
\begin{proof}[Proof of Theorem~\ref{thm_su2p}]
    Given an instance $\phi$ of $(3,3)$-SAT with $n$ variables and $m$ clauses, we construct the graph $G_\phi$ (an illustrative example shown in Figure~\ref{fig_SU2P_emp}) by assembling variable gadgets, clause gadgets, and the necessary constant neighbours as described earlier in the section. By Propositions~\ref{prop_SU2P_d3_cst}, ~\ref{prop_SU2P_d3_var}, and ~\ref{prop_SU2P_d3_3cls}, it follows that $G_\phi$ admits a strictly unfriendly 2-partition if and only if $\phi$ is satisfiable, giving SU2P is NP-hard.

    The total number of vertices $N$ in $G_\phi$ is bounded as follows:
    \begin{itemize}
        \item Variables: $n$ variable gadgets $\times$ 4 vertices/gadget = $4n$.
        \item Clauses: $l\leq3n$ literals $\times$ 3 vertices/literal (worst-case for 3-clauses) $\leq 9n$.
    \end{itemize}
    Thus, $N \leq 4n + 9n = 13n$. By the linear nature of the reduction, any algorithm solving SU2P in $2^{o(N)}$ time would imply a $2^{o(n)}$ time algorithm for $(3,3)$-SAT, contradicting ETH. Specifically:
     $$c_{(3,3)\text{-SAT}}^n \leq c_{\text{SU2P}}^{13n} \implies c_{\text{SU2P}} \geq c_{(3,3)\text{-SAT}}^{1/13}.$$
    This completes the proof.
\end{proof}

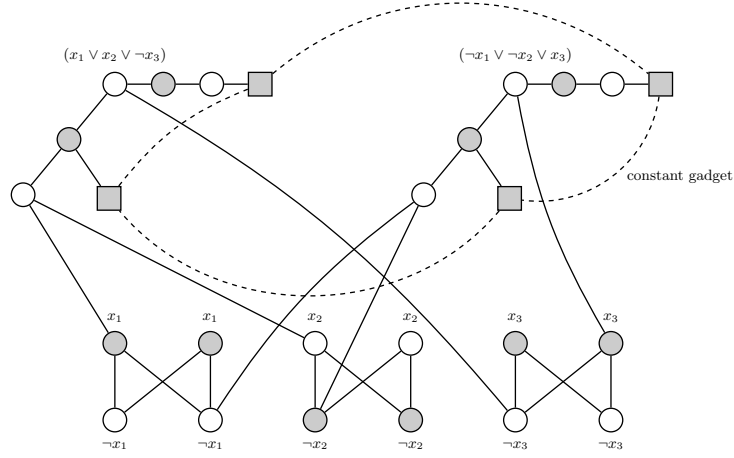
\begin{figure}
    \centering
    \resizebox{0.8\textwidth}{!}{
    \begin{tikzpicture}[
    v_node/.style={circle, draw, minimum size=5mm, thick, inner sep=0pt},
    u_node/.style={rectangle, draw, minimum size=5mm, thick, inner sep=0pt},
    fill_node/.style={fill=gray!40}, 
    edge/.style={draw, thick},
    link/.style={draw, thick},
    cycle_link/.style={draw, black, thick, dashed},
    cycle_link_2/.style={draw, black, thick}
]

    \foreach \x / \name in {0/a1, 8.5/a2} { 
        \begin{scope}[xshift=\x cm]
            \ifnum\pdfstrcmp{\name}{a1}=0
                \node[v_node] (\name-v1) {};
                \node[v_node, fill_node] (\name-v2) [right=0.5cm of \name-v1] {};
                \node[v_node] (\name-v3) [right=0.5cm of \name-v2] {};
                \node[u_node, fill_node] (\name-u2) [right=0.5cm of \name-v3] {};
                \node[v_node, fill_node](\name-v4) [below left=0.8cm and 0.6cm of \name-v1] {};
                \node[v_node] (\name-v5) [below left=0.8cm and 0.6cm of \name-v4] {};
                \node[u_node, fill_node] (\name-u1) [below right=0.8cm and 0.4cm of \name-v4] {};
                
                \node[above=2pt of \name-v1] {$(x_1 \lor x_2 \lor \neg x_3)$};
            \else
                \node[v_node] (\name-v1) {};
                \node[v_node, fill_node] (\name-v2) [right=0.5cm of \name-v1] {};
                \node[v_node] (\name-v3) [right=0.5cm of \name-v2] {};
                \node[u_node, fill_node] (\name-u2) [right=0.5cm of \name-v3] {};
                \node[v_node, fill_node] (\name-v4) [below left=0.8cm and 0.6cm of \name-v1] {};
                \node[v_node] (\name-v5) [below left=0.8cm and 0.6cm of \name-v4] {};
                \node[u_node, fill_node] (\name-u1) [below right=0.8cm and 0.4cm of \name-v4] {};
                
                \node[above=2pt of \name-v1] {$(\neg x_1 \lor \neg x_2 \lor x_3)$};
            \fi

            \draw[edge] (\name-v1) -- (\name-v2); 
            \draw[edge] (\name-v2) -- (\name-v3); 
            \draw[edge] (\name-v3) -- (\name-u2);
            \draw[edge] (\name-v1) -- (\name-v4); 
            \draw[edge] (\name-v4) -- (\name-v5); 
            \draw[edge] (\name-v4) -- (\name-u1);
        \end{scope}
    }

    \foreach \x / \name / \i in {0/b1/1, 4.25/b2/2, 8.5/b3/3} {
        \begin{scope}[xshift=\x cm, yshift=-5.5cm]
            \ifnum\i=2
                \node[v_node] (\name-v1) {};
                \node[v_node, fill_node] (\name-v3) [below=1.1cm of \name-v1] {};
                \node[v_node] (\name-v2) [right=1.5cm of \name-v1] {};
                \node[v_node, fill_node] (\name-v4) [below=1.1cm of \name-v2] {};
            \else
                \node[v_node, fill_node] (\name-v1) {};
                \node[v_node] (\name-v3) [below=1.1cm of \name-v1] {};
                \node[v_node, fill_node] (\name-v2) [right=1.5cm of \name-v1] {};
                \node[v_node] (\name-v4) [below=1.1cm of \name-v2] {};
            \fi

            \node[above=2pt of \name-v1] {$x_{\i}$};
            \node[above=2pt of \name-v2] {$x_{\i}$};
            \node[below=2pt of \name-v3] {$\neg x_{\i}$};
            \node[below=2pt of \name-v4] {$\neg x_{\i}$};

            \draw[edge] (\name-v1) -- (\name-v3); \draw[edge] (\name-v2) -- (\name-v4);
            \draw[edge] (\name-v1) -- (\name-v4); \draw[edge] (\name-v3) -- (\name-v2);
        \end{scope}
    }

    \draw[link] (b1-v1) -- (a1-v5);
    \draw[link] (b2-v1) -- (a1-v5);
    \draw[cycle_link_2] (b3-v3) to[bend right=10] (a1-v1);
    
    \draw[cycle_link_2] (b1-v4) to[bend left=10] (a2-v5);
    \draw[link] (b2-v3) -- (a2-v5); 
    \draw[cycle_link_2] (b3-v2) to[bend left=10] (a2-v1);

    \draw[cycle_link] (a1-u1) to[bend left=15] (a1-u2);
    \draw[cycle_link] (a1-u2) to[out=40, in=140] (a2-u2);
    
    \draw[cycle_link] (a2-u2) to[bend left=45] node[midway, right=0.2cm] {constant gadget} (a2-u1);
    
    \draw[cycle_link] (a2-u1) to[out=230, in=310] (a1-u1);

\end{tikzpicture}
}
    \caption{A legal SU2P on $G_\phi$, where $\phi=(x_1\lor x_2\lor\lnot x_3)\land(\lnot x_1\lor \lnot x_2\lor x_3)$ and the truth assignment $(x_1,x_2,x_3)=(1,0,1)$. Grey vertices have label 1; blank vertices have label 0.}
    \label{fig_SU2P_emp}
\end{figure}

\section{Hardness of Strictly 3-Partition (SU3P)}
\label{sec_su3p}

The complexity analysis of SU3P is similar to that of SU2P, yet the shift to three partition introduces significantly more stringent local constraints. In a 3-partition, for a vertex to be legal, its neighbours must be partitioned such that each of the other two partition classes contains strictly more neighbours than the vertex's own partition class. Similar to the $k=2$ case, SU3P is efficiently solvable on graphs with maximum degree two (where it requires the neighbours of each vertex to belong to distinct complementary partition classes). In this section, we prove that SU3P becomes NP-hard as soon as the maximum degree reaches three, thereby pinpointing the degree-based complexity transition.
\begin{theorem} \label{thm_SU3P}
    SU3P is NP-hard on graphs with maximum degree three. Assuming ETH, it cannot be solved in $2^{o(n)}$ time; specifically, $c_{(3,3)\text{-SAT}}^{1/38} \leq c_{\text{SU3P}}$.
\end{theorem}
The reduction from $(3,3)$-SAT to SU3P maintains the modular architecture of variable, clause, and constant gadgets. However, the gadgets must be redesigned to accommodate the unique requirements of a 3-partition. As the reduction remains linear in terms of the number of variables, the ETH-based lower bound follows directly from Lemma~\ref{lem_33}.

\medskip
\noindent \textbf{Key Observations for SU3P.}
To simplify the reduction and the subsequent proof of correctness, we formalize the local constraints imposed by the strictly unfriendly 3-partition condition on low-degree vertices. Throughout this section, the label set is $L=\{0,1,2\}$. The following observations are straightforward; their formal proofs are deferred to appendix (Section~\ref{sec_app_prf3}).

\begin{observation}[Zero-neighbour Requirement] \label{ob_SU3P_zero}
    In any legal SU3P of a graph with maximum degree $\Delta=3$, every vertex $v$ must satisfy $d_{\pi(v)}(v) = 0$.
\end{observation}

\begin{observation}[Constraints on Degree-2 Vertices] \label{ob_SU3P_d2}
    In any legal SU3P, a degree-2 vertex $u$ must have its two neighbours having the two distinct labels in $L \setminus \{\pi(u)\}$.
\end{observation}

\begin{observation}[Constraints on Degree-3 Vertices] \label{ob_SU3P_d3}
    In any legal SU3P, a degree-3 vertex $u$ must have its three neighbours split between the two labels other than $\pi(u)$ in a $(2, 1)$ or $(1, 2)$ distribution.
\end{observation}

\begin{observation}[Label Propagation via Degree-2 Paths] \label{ob_SU3P_ppg}
    In any legal SU3P, if two vertices $v_1, v_2$ are connected by a path of two degree-2 vertices $(w_1, w_2)$, then $v_1$ and $v_2$ must have the same label.
\end{observation}

\medskip
\noindent \textbf{Reduction Overview.}
The reduction paradigm mirrors that of SU2P, utilizing three categories of gadgets: constant, variable, and clause. However, the construction in the $k=3$ case requires two distinct constant signals, namely \emph{constant-2} and \emph{constant-0}, to restrict the literal vertices to the truth-value set $\{0, 1\}$. These dual signals can be generated by coupling two constant cycles, effectively providing the necessary reference labels for the rest of the construction. Proofs of the propositions are deferred to appendix (Section~\ref{sec_app_prf3}).

\medskip
\noindent \textbf{Constant Gadgets.} To provide stable truth-value gadgets in the SU3P reduction, we define a generalized \emph{Forcing Cycle} that propagates a fixed label to multiple output vertices. A Forcing Cycle $C_{3k}$ (Figure~\ref{fig_SU3P_d3_cst2} in Appendix) consists of $3k$ vertices $\{u_1, u_2, \dots, u_{3k}\}$ organized into functional blocks of three. In the following content, all vertex subscripts are taken modulo $3k$ with representative indices in $\{1, \dots, 3k\}$, such that $u_0$ is identified with $u_{3k}$. For each $j \in \{0, \dots, k-1\}$:
\begin{itemize}
    \item Anchor vertex: $u_{3j+1}$ is a degree-3 vertex. It serves as an interface to either an external forcing signal or an output literal.
    \item Propagation vertices: $u_{3j+2}$ and $u_{3j+3}$ are degree-2 vertices that connect $u_{3j+1}$ to the next anchor vertex $u_{3j+4}$.
\end{itemize}
On average, each forced constant vertex requires three vertices ($u_{3j+1}, u_{3j+2}, u_{3j+3}$), resulting in an amortized cost of 3 vertices per signal.

\begin{proposition} \label{prop_SU3P_d3_cst}
    In any legal strictly unfriendly 3-partition $\pi$ of a Forcing Cycle, all anchor vertices must share the same label, i.e., $\pi(u_{3j+1}) = \pi(u_{3j+4})$ for all $j$.
\end{proposition}

We now use Forcing Cycles to generate constant-2 and constant-0 gadgets. In any legal 3-partition, the specific choice of labels is subject to global permutation symmetry. We exploit this to define our constant gadgets:
\begin{enumerate}
    \item Constant-2 Gadget: We designate the anchor vertices of the first Forcing Cycle as the source for label $2$. By Proposition~\ref{prop_SU3P_d3_cst}, all anchor vertices in this cycle are forced to share this same label, which we denote as constant-2 without loss of generality.
    \item Constant-0 Gadget: We take a second Forcing Cycle and connect its first anchor vertex $u_1'$ to a constant-2 anchor vertex from the first cycle. This connection forces $\pi(u_1') \neq 2$. By the remaining symmetry between labels $0$ and $1$, we can fix $\pi(u_1') = 0$. Consequently, all anchor vertices in the second cycle are forced to be constant-0.
\end{enumerate}
This modular approach provides a stable supply of constant-0 and 2 neighbours for the variable and clause gadgets, with each maintaining an amortized cost of 3 vertices.

\medskip
\noindent \textbf{Variable Gadget.} The variable gadget for $x$ is designed to represent truth assignments within the Boolean domain $\{0, 1\}$, while ensuring consistency and polarity between occurrences of $x$ and $\neg x$. As illustrated in Figure~\ref{fig_SU3P_d3_var}, the construction employs eight internal vertices $v_1, \dots, v_8$ and two constant-2 signals provided by anchor vertices $u, u'$. The vertices are categorized as follows:
\begin{itemize}
    \item Literal vertices: $v_1, v_4$ (representing $x$) and $v_5, v_8$ (representing $\neg x$). These vertices serve as interfaces to clause gadgets.
    \item Propagation paths: The pairs $(v_2, v_3)$ and $(v_6, v_7)$ are degree-2 vertices acting as synchronization buffers between literal occurrences.
    \item Domain restriction: Vertices $v_4$ and $v_5$ are connected to constant-2 anchors to exclude label $2$ from the truth-value assignments.
\end{itemize}
Each variable gadget utilizes 8 internal vertices and 6 amortized vertices from the constant-2 forcing cycles (2 signals $\times$ 3 vertices), totaling 14 vertices per variable.

\begin{figure}
    \centering
    \begin{subfigure}[b]{0.4\textwidth}
        \centering
    \resizebox{0.99\textwidth}{!}{
    \begin{tikzpicture}[
normal/.style = {draw, regular polygon,regular polygon sides=360, very thick},
fixed/.style={draw, regular polygon,regular polygon sides=4, very thick, minimum size=1.2cm},
balance/.style = {draw, regular polygon,regular polygon sides=360, fill=black},
]
\node[fixed] (u1) [label=above:$2$] {$u$};
\node[fixed] (u2) [right=1cm of u1, label=above:$2$] {$u'$};
\node[normal] (v4) [below=1cm of u1, label=below:$x$] {$v_4$};
\node[normal] (v5) [below=1cm of u2, label=below:$\lnot x$] {$v_5$};
\node[normal] (v3) [left=1cm of v4] {$v_3$};
\node[normal] (v2) [below=1cm of v3] {$v_2$};
\node[normal] (v1) [right=1cm of v2, label=below:$x$] {$v_1$};
\node[] (u3) [below left=1cm of v1] {to clause};
\node[] (u4) [below right=1cm of v1] {to clause};
\node[normal] (v6) [right=1cm of v5] {$v_6$};
\node[normal] (v7) [below=1cm of v6] {$v_7$};
\node[normal] (v8) [left=1cm of v7, label=below:$\lnot x$] {$v_8$};
\node[] (u5) [below right=1cm of v8] {to clause};

\draw (u1) -- (v4);
\draw (u2) -- (v5);
\draw (v3) -- (v4);
\draw (v3) -- (v2);
\draw (v2) -- (v1);
\draw (u3) -- (v1);
\draw (u4) -- (v1);
\draw (v6) -- (v5);
\draw (v6) -- (v7);
\draw (v8) -- (v7);
\draw (v8) -- (u5);
\draw (v4) -- (v5);

\end{tikzpicture}
}
    \caption{Labels $\pi(v)$ are indicated above the vertices.}
    \label{fig_SU3P_d3_var}
    \end{subfigure}
    \hfill
    \begin{subfigure}[b]{0.45\textwidth}
        \centering
        \resizebox{0.99\textwidth}{!}{
    \begin{tikzpicture}[
normal/.style = {draw, regular polygon,regular polygon sides=360, very thick},
fixed/.style={draw, regular polygon,regular polygon sides=4, very thick, minimum size=1.2cm},
balance/.style = {draw, regular polygon,regular polygon sides=360, fill=black},
]
\node[fixed] (u1) [label=above:$2$] {$u$};
\node[fixed] (u2) [right=1cm of u1, label=above:$2$] {$u'$};
\node[normal] (v1) [below=1cm of u1, label=above left:$c$] {$v_1$};
\node[normal] (v2) [below=1cm of u2, label=above right:$\lnot c$] {$v_2$};
\node[normal] (v3) [below left=1cm of v1, label=above:$2$] {$v_3$};
\node[] (v5) [below left=1cm of v3] {$x$};
\node[] (v6) [below right=1cm of v3] {$y$};
\node[normal] (v4) [below right=1cm of v2, label=above:$2$] {$v_4$};
\node[] (v7) [below left=1cm of v4] {$z$};
\node[fixed] (u3) [below right=1cm of v4, label=below:$0$] {$u''$};

\draw (u1) -- (v1);
\draw (u2) -- (v2);
\draw (v1) -- (v2);
\draw (v1) -- (v3);
\draw (v3) -- (v5);
\draw (v3) -- (v6);
\draw (v4) -- (v2);
\draw (v4) -- (v7);
\draw (v4) -- (u3);

\end{tikzpicture}
}
    \caption{Labels $\pi(v)$ are indicated above the vertices.}
    \label{fig_SU3P_d3_3cls}
    \end{subfigure}
    \caption{}
\end{figure}

\begin{proposition}
    \label{prop_SU3P_d3_var}
    A variable gadget is legally partitioned if and only if the literal vertices satisfy $\pi(v_1) = \pi(v_4)$, $\pi(v_5) = \pi(v_8)$, and $\{\pi(v_4), \pi(v_5)\} = \{0, 1\}$.
\end{proposition}

\medskip
\noindent \textbf{Clause Gadget.} The gadget for a 3-literal clause $(x \lor y \lor z)$ is designed to enforce satisfiability by leveraging Observation~\ref{ob_SU3P_zero}. As illustrated in Figure~\ref{fig_SU3P_d3_3cls}, the construction employs four internal vertices $v_1, \dots, v_4$ and three constant signals: two constant-2 anchors ($u, u'$) and one constant-0 anchor ($u''$). Each 3-literal clause introduces 4 internal vertices and 9 amortized vertices from the forcing cycles (3 signals $\times$ 3 vertices), totaling 13 vertices, and thus $13/3$ vertices per literal.



\begin{proposition} \label{prop_SU3P_d3_3cls}
    A clause gadget is legally partitioned if and only if at least one of the literal vertices $\{x, y, z\}$ is assigned the label $1$ (True).
\end{proposition}

The gadget for a 2-literal clause $(x \lor z)$ is derived by replacing the literal $y$ in the 3-clause construction with an additional constant-0 anchor $u'''$ (shown in Appendix Figure~\ref{fig_SU3P_d3_2cls}). This modification has a vertex count for 2-clauses: each such clause requires 4 internal vertices and 12 amortized vertices from 4 constant signals ($u, u', u'', u'''$), totaling $16$ vertices, and thus $16/2=8$ vertices per literal. Compared with 3-clause gadget with $13/3$ vertices per literal, 2-clause gadget uses more vertices. To ensure the robustness of our ETH-based lower bound, we utilize the worst case (8 vertices per literal, as the more the vertices are used through the reduction, the less the ETH-based bound is) in our final calculation.

\medskip
\noindent \textbf{Full Reduction.} With the gadgets defined, we now summarize the reduction and analyze the resulting ETH-based lower bound.
\begin{proof}[Proof of Theorem~\ref{thm_SU3P}]
    Given an instance $\phi$ of $(3,3)$-SAT with $n$ variables and $m$ clauses, we construct the graph $G_\phi$ (an illustrative example shown in Figure~\ref{fig_SU3P_emp}) by assembling variable gadgets, clause gadgets, and the necessary constant neighbours as described earlier in the section. By Propositions~\ref{prop_SU3P_d3_cst}, ~\ref{prop_SU3P_d3_var}, and ~\ref{prop_SU3P_d3_3cls}, it follows that $G_\phi$ admits a strictly unfriendly 3-partition if and only if $\phi$ is satisfiable, giving SU3P is NP-hard.
    
    The total number of vertices $N$ in $G_\phi$ is bounded as follows:
    \begin{itemize}
        \item Variables: $n$ variable gadgets $\times$ 14 vertices/gadget = $14n$.
        \item Clauses: $l\leq3n$ literals $\times$ 8 vertices/literal (worst-case for 2-clauses) $\leq 24n$.
    \end{itemize}
    Thus, $N \leq 14n + 24n = 38n$. By the linear nature of the reduction, any algorithm solving SU3P in $2^{o(N)}$ time would imply a $2^{o(n)}$ time algorithm for $(3,3)$-SAT, contradicting ETH. Specifically:
     $$c_{(3,3)\text{-SAT}}^n \leq c_{\text{SU3P}}^{38n} \implies c_{\text{SU3P}} \geq c_{(3,3)\text{-SAT}}^{1/38}.$$
    This completes the proof.
\end{proof}

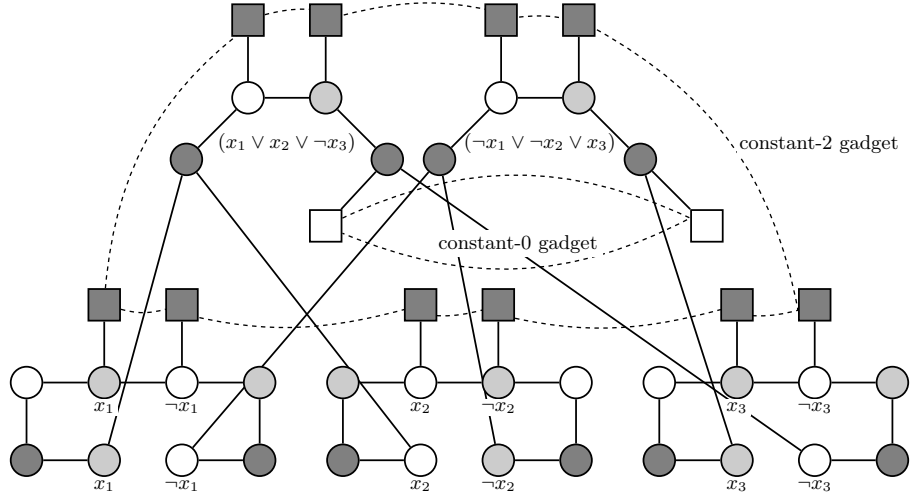
\begin{figure}[H]
    \centering
    \resizebox{0.99\textwidth}{!}{
    \begin{tikzpicture}[
    v_node/.style={circle, draw, minimum size=5mm, thick, inner sep=0pt, fill=white},
    u_node/.style={rectangle, draw, minimum size=5mm, thick, inner sep=0pt, font=\small, fill=white},
    filled_light/.style={fill=gray!40},
    filled_gray/.style={fill=gray},
    edge/.style={draw, thick},
    dense_dash/.style={draw, dashed, dash pattern=on 2pt off 2pt, semithick},
    extra_dense_dash/.style={draw, dashed, dash pattern=on 1.5pt off 1.5pt, semithick},
    lbl/.style={font=\small, fill=white, inner sep=1.2pt}
]

    \begin{scope}[xshift=8.5cm, yshift=0cm]
        \node[v_node] (d2-v1) {};
        \node[v_node, filled_light] (d2-v2) [right=0.7cm of d2-v1] {};
        \node[u_node, filled_gray] (d2-u1) [above=0.7cm of d2-v1] {};
        \node[u_node, filled_gray] (d2-u2) [above=0.7cm of d2-v2] {};
        \node[v_node, filled_gray] (d2-v3) [below left=0.6cm and 0.6cm of d2-v1] {};
        \node[v_node, filled_gray] (d2-v4) [below right=0.6cm and 0.6cm of d2-v2] {};
        \node[u_node] (d2-u3) [below right=0.6cm and 0.6cm of d2-v4] {};
        
        \draw[edge] (d2-v1) -- node[lbl, midway, below=0.5cm] {$(\neg x_1 \lor \neg x_2 \lor x_3)$} (d2-v2);
        \draw[edge] (d2-u1) -- (d2-v1);
        \draw[edge] (d2-u2) -- (d2-v2);
        \draw[edge] (d2-v1) -- (d2-v3);
        \draw[edge] (d2-v2) -- (d2-v4);
        \draw[edge] (d2-v4) -- (d2-u3);
    \end{scope}

    \begin{scope}[xshift=4.5cm, yshift=0cm]
        \node[v_node] (d1-v1) {};
        \node[v_node, filled_light] (d1-v2) [right=0.7cm of d1-v1] {};
        \node[u_node, filled_gray] (d1-u1) [above=0.7cm of d1-v1] {};
        \node[u_node, filled_gray] (d1-u2) [above=0.7cm of d1-v2] {};
        \node[v_node, filled_gray] (d1-v3) [below left=0.6cm and 0.6cm of d1-v1] {};
        \node[v_node, filled_gray] (d1-v4) [below right=0.6cm and 0.6cm of d1-v2] {};
        \node[u_node] (d1-u3) at (d1-v2 |- d2-u3) {};

        \draw[edge] (d1-v1) -- node[lbl, midway, below=0.5cm] {$(x_1 \lor x_2 \lor \neg x_3)$} (d1-v2);
        \draw[edge] (d1-u1) -- (d1-v1);
        \draw[edge] (d1-u2) -- (d1-v2);
        \draw[edge] (d1-v1) -- (d1-v3);
        \draw[edge] (d1-v2) -- (d1-v4);
        \draw[edge] (d1-v4) -- (d1-u3);
    \end{scope}

    
    \begin{scope}[xshift=1.0cm, yshift=-4.5cm]
        \node[v_node] (c1-v3) {};
        \node[v_node, filled_light] (c1-v4) [right=0.7cm of c1-v3, label={[lbl]below:$x_{1}$}] {};
        \node[v_node] (c1-v5) [right=0.7cm of c1-v4, label={[lbl]below:$\neg x_{1}$}] {};
        \node[v_node, filled_light] (c1-v6) [right=0.7cm of c1-v5] {};
        \node[u_node, filled_gray] (c1-u1) [above=0.7cm of c1-v4] {};
        \node[u_node, filled_gray] (c1-u2) [above=0.7cm of c1-v5] {};
        \node[v_node, filled_gray] (c1-v2) [below=0.7cm of c1-v3] {};
        \node[v_node, filled_gray] (c1-v7) [below=0.7cm of c1-v6] {};
        \node[v_node, filled_light] (c1-v1) [right=0.7cm of c1-v2, label={[lbl]below:$x_{1}$}] {};
        \node[v_node] (c1-v8) [left=0.7cm of c1-v7, label={[lbl]below:$\neg x_{1}$}] {};

        \draw[edge] (c1-v3)--(c1-v4)--(c1-v5)--(c1-v6) (c1-u1)--(c1-v4) (c1-u2)--(c1-v5) 
                   (c1-v3)--(c1-v2)--(c1-v1) (c1-v6)--(c1-v7)--(c1-v8);
    \end{scope}

    \begin{scope}[xshift=6.0cm, yshift=-4.5cm]
        \node[v_node, filled_light] (c2-v3) {};
        \node[v_node] (c2-v4) [right=0.7cm of c2-v3, label={[lbl]below:$x_{2}$}] {};
        \node[v_node, filled_light] (c2-v5) [right=0.7cm of c2-v4, label={[lbl]below:$\neg x_{2}$}] {};
        \node[v_node] (c2-v6) [right=0.7cm of c2-v5] {};
        \node[u_node, filled_gray] (c2-u1) [above=0.7cm of c2-v4] {};
        \node[u_node, filled_gray] (c2-u2) [above=0.7cm of c2-v5] {};
        \node[v_node, filled_gray] (c2-v2) [below=0.7cm of c2-v3] {};
        \node[v_node, filled_gray] (c2-v7) [below=0.7cm of c2-v6] {};
        \node[v_node] (c2-v1) [right=0.7cm of c2-v2, label={[lbl]below:$x_{2}$}] {};
        \node[v_node, filled_light] (c2-v8) [left=0.7cm of c2-v7, label={[lbl]below:$\neg x_{2}$}] {};

        \draw[edge] (c2-v3)--(c2-v4)--(c2-v5)--(c2-v6) (c2-u1)--(c2-v4) (c2-u2)--(c2-v5) 
                   (c2-v3)--(c2-v2)--(c2-v1) (c2-v6)--(c2-v7)--(c2-v8);
    \end{scope}

    \begin{scope}[xshift=11.0cm, yshift=-4.5cm]
        \node[v_node] (c3-v3) {};
        \node[v_node, filled_light] (c3-v4) [right=0.7cm of c3-v3, label={[lbl]below:$x_{3}$}] {};
        \node[v_node] (c3-v5) [right=0.7cm of c3-v4, label={[lbl]below:$\neg x_{3}$}] {};
        \node[v_node, filled_light] (c3-v6) [right=0.7cm of c3-v5] {};
        \node[u_node, filled_gray] (c3-u1) [above=0.7cm of c3-v4] {};
        \node[u_node, filled_gray] (c3-u2) [above=0.7cm of c3-v5] {};
        \node[v_node, filled_gray] (c3-v2) [below=0.7cm of c3-v3] {};
        \node[v_node, filled_gray] (c3-v7) [below=0.7cm of c3-v6] {};
        \node[v_node, filled_light] (c3-v1) [right=0.7cm of c3-v2, label={[lbl]below:$x_{3}$}] {};
        \node[v_node] (c3-v8) [left=0.7cm of c3-v7, label={[lbl]below:$\neg x_{3}$}] {};

        \draw[edge] (c3-v3)--(c3-v4)--(c3-v5)--(c3-v6) (c3-u1)--(c3-v4) (c3-u2)--(c3-v5) 
                   (c3-v3)--(c3-v2)--(c3-v1) (c3-v6)--(c3-v7)--(c3-v8);
    \end{scope}

    \begin{pgfonlayer}{background}
        \draw[edge] (d1-v3) -- (c1-v1);
        \draw[edge] (d1-v3) -- (c2-v1);
        \draw[edge] (d1-v4) -- (c3-v8);
        \draw[edge] (d2-v3) -- (c1-v8);
        \draw[edge] (d2-v3) -- (c2-v8);
        \draw[edge] (d2-v4) -- (c3-v1);

        \path [dense_dash] (c1-u1) to [bend right=15] (c1-u2)
                                 to [bend right=15] (c2-u1)
                                 to [bend right=15] (c2-u2)
                                 to [bend right=15] (c3-u1)
                                 to [bend right=15] (c3-u2) 
                                 to [bend right=25] node[lbl, pos=0.5, right=2pt] {constant-2 gadget} (d2-u2)
                                 to [bend right=15] (d2-u1)
                                 to [bend right=15] (d1-u2)
                                 to [bend right=15] (d1-u1)
                                 to [bend right=25] (c1-u1);
                                 
        \path [dense_dash] (d1-u3) to [bend left=25] (d2-u3)
                                   to [bend left=25] node[lbl, midway, above=5pt, fill=white, inner sep=2pt] {constant-0 gadget} (d1-u3);
    \end{pgfonlayer}

\end{tikzpicture}
}
    \caption{A legal SU3P on $G_\phi$, where $\phi=(x_1\lor x_2\lor\lnot x_3)\land(\lnot x_1\lor \lnot x_2\lor x_3)$ and the truth assignment $(x_1,x_2,x_3)=(1,0,1)$. Dark vertices have label 2; grey vertices have label 1; blank vertices have label 0.}
    \label{fig_SU3P_emp}
\end{figure}

\section{General Hardness for $k \geq 2$}
\label{sec_general_hardness}

The modular reduction framework generalizes to establish the $\mathbf{NP}$-hardness and ETH-hardness of $\text{SU}k\text{P}$ for any fixed $k \geq 2$ on general graphs.

\begin{theorem} \label{thm_SUkP_unrestricted}
For any fixed $k \geq 2$, $\text{SU}k\text{P}$ on general graphs is $\mathbf{NP}$-hard~\cite{kun2013anti}. Assuming ETH, $c_{\text{SU}k\text{P}} \geq c_{(3,3)\text{-SAT}}^{2/7}$.
\end{theorem}

\medskip
\noindent \textbf{Technical Divergence and Complexity Invariance.}
Despite technical differences between parity-constrained $\text{SU}2\text{P}$ gadgets and the flexible multi-partition framework for $k \geq 3$, both reductions converge to an identical lower bound constant of $c_{(3,3)\text{-SAT}}^{2/7}$. This demonstrates that the core exponential complexity on unconstrained graphs remains invariant across all $k \ge 2$.

\medskip
\noindent \textbf{The Efficiency-Degree Trade-off.}
Enforcing subcubic constraints ($\Delta \le 3$) requires elaborate synchronization gadgets, inflating instance sizes to $N \le 13n$ ($k=2$) and $N \le 38n$ ($k=3$). Relaxing degree constraints enables compact uniform gadgets with $N = 3.5n$ for all $k \geq 2$, yielding the tighter lower bound of $c_{(3,3)\text{-SAT}}^{2/7}$. This formally quantifies the efficiency loss incurred by graph sparsity.

\section{Conclusion}
\label{sec_con}

In this paper, we provided a comprehensive complexity analysis of the Strictly Unfriendly $k$-Partition problem. In summary, while the $2^{\Omega(n)}$ exponential lower bound of $\text{SU}k\text{P}$ is robust across $k$, the specific exponential constant is highly sensitive to structural constraints, particularly the maximum degree. Future research may determine if the $2/7$ bound is optimal for general graphs or if further gadget optimization can refine this constant. Our main results are summarized as follows:

\begin{itemize}
    \item \textbf{Tractability Thresholds:} It is easy to see that $\Delta \geq k-1$ as a necessary condition for $\text{SU}k\text{P}$ existence, rendering the problem trivial when this bound is not met. For the pivotal cases of $k \in \{2, 3\}$, we pinpointed a sharp complexity transition at $\Delta = 3$: while these problems remain efficiently solvable for $\Delta \leq 2$, they become $\mathbf{NP}$-hard and ETH-hard on subcubic graphs.
    
    \item \textbf{Complexity Convergence:} Our results reveal a striking \emph{complexity convergence} for $\text{SU}k\text{P}$ on general graphs. Despite the notable technical divergence in gadget constructions for bipartitions versus multi-partitions, both reductions converge to an identical lower bound constant of $c_{(3,3)\text{-SAT}}^{2/7}$. This suggests that the fundamental complexity essence of the strictly unfriendly condition remains invariant across all $k \geq 2$ when degree constraints are relaxed.
    
    \item \textbf{The Cost of Sparsity:} We quantified the inherent trade-off between graph sparsity and reduction efficiency. Our results demonstrate that restricting gadgets to subcubic graphs necessitates a significantly larger reduction size, which is reflected in the sharp decrease of the exponents—from $2/7$ on general graphs to $1/13$ for $\text{SU}2\text{P}$ and $1/38$ for $\text{SU}3\text{P}$ on subcubic instances.
\end{itemize}

\bibliography{TAMC_full/ref}

\section{Appendix: Omitted Proofs for SUP}
\label{sec_app}

\subsection{Omitted Figure}

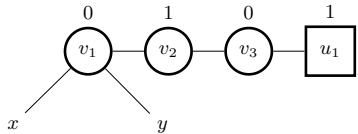
\begin{figure}[H]
    \centering
    \resizebox{0.4\textwidth}{!}{
    \begin{tikzpicture}[
normal/.style = {draw, regular polygon,regular polygon sides=360, very thick},
fixed/.style={draw, regular polygon,regular polygon sides=4, very thick},
balance/.style = {draw, regular polygon,regular polygon sides=360, fill=black},
]
\node[normal] (v1) [label=above:$0$] {$v_1$};
\node[normal] (v2) [right=0.5cm of v1, label=above:$1$] {$v_2$};
\node[normal] (v3) [right=0.5cm of v2, label=above:$0$] {$v_3$};
\node[fixed] (u1) [right=0.5cm of v3, label=above:$1$] {$u_1$};
\node[] (v4) [below left=1cm of v1] {$x$};
\node[] (v5) [below right=1cm of v1] {$y$};

\draw (v3) -- (u1);
\draw (v3) -- (v2);
\draw (v1) -- (v2);
\draw (v1) -- (v4);
\draw (v1) -- (v5);

\end{tikzpicture}
}
    \caption{Labels $\pi(v)$ are indicated above the vertices.}
    \label{fig_SU2P_d3_2cls}
\end{figure}

\begin{figure}[H]
    \centering
    \begin{subfigure}[b]{0.49\textwidth}
        \centering
    \resizebox{0.66\textwidth}{!}{
    \begin{tikzpicture}[
        normal/.style = {draw, regular polygon,regular polygon sides=360, very thick},
        fixed/.style={draw, regular polygon,regular polygon sides=4, very thick, minimum size=2.4cm},
        balance/.style = {draw, regular polygon,regular polygon sides=360, fill=black},
    ]
        \node[fixed] (u1) {$u_1$};
        \node[fixed] (u2) [below=1cm of u1] {$u_2$};
        \node[fixed] (u3) [below=1cm of u2] {$u_3$};
        \node[fixed] (u4) [left=1cm of u3] {$u_4$};
        \node[fixed] (u2k-1) [above=1cm of u4] {$u_{2k-1}$};
        \node[fixed] (u2k) [above=1cm of u2k-1] {$u_{2k}$};
        \node[] (e1) [above=1cm of u1] {};
        \node[] (e3) [below=1cm of u3] {};
        \node[] (e2k-1) [left=1cm of u2k-1] {};

        \draw (u1) -- (u2);
        \draw (u3) -- (u2);
        \draw (u3) -- (u4);
        \draw [very thick, dotted] (u2k-1) -- (u4);
        \draw (u2k-1) -- (u2k);
        \draw (u1) -- (u2k);
        \draw (u1) -- (e1);
        \draw (u3) -- (e3);
        \draw (u2k-1) -- (e2k-1);

    \end{tikzpicture}
    }
    \caption{}
    \label{fig_SU2P_d3_cst}
    \end{subfigure}
    \hfill
    \begin{subfigure}[b]{0.49\textwidth}
        \resizebox{0.66\textwidth}{!}{
    \begin{tikzpicture}[
        normal/.style = {draw, regular polygon,regular polygon sides=360, very thick},
        fixed/.style={draw, regular polygon,regular polygon sides=4, very thick, minimum size=2.4cm},
        balance/.style = {draw, regular polygon,regular polygon sides=360, fill=black},
    ]
        \node[fixed] (u1) {$u_1$};
        \node[fixed] (u2) [below=1cm of u1] {$u_2$};
        \node[fixed] (u3) [below=1cm of u2] {$u_3$};
        \node[fixed] (u4) [below=1cm of u3] {$u_4$};
        \node[fixed] (u5) [left=1cm of u4] {$u_5$};
        \node[fixed] (u3k-2) [above=1cm of u5] {$u_{3k-2}$};
        \node[fixed] (u3k-1) [above=1cm of u3k-2] {$u_{3k-1}$};
        \node[fixed] (u3k) [above=1cm of u3k-1] {$u_{3k}$};
        \node[] (e1) [above=1cm of u1] {};
        \node[] (e4) [below=1cm of u4] {};
        \node[] (e3k-2) [left=1cm of u3k-2] {};

        \draw (u1) -- (u2);
        \draw (u3) -- (u2);
        \draw (u3) -- (u4);
        \draw (u5) -- (u4);
        \draw [dotted, very thick] (u5) -- (u3k-2);
        \draw (u3k-1) -- (u3k-2);
        \draw (u3k-1) -- (u3k);
        \draw (u1) -- (u3k);
        \draw (u1) -- (e1);
        \draw (u4) -- (e4);
        \draw (u3k-2) -- (e3k-2);

    \end{tikzpicture}
    }
    \caption{}
    \label{fig_SU3P_d3_cst2}
    \end{subfigure}
    \caption{}
    \label{fig:placeholder}
\end{figure}

\begin{figure}[H]
    \centering
    \resizebox{0.45\textwidth}{!}{
    \begin{tikzpicture}[
normal/.style = {draw, regular polygon,regular polygon sides=360, very thick},
fixed/.style={draw, regular polygon,regular polygon sides=4, very thick},
balance/.style = {draw, regular polygon,regular polygon sides=360, fill=black},
]
\node[fixed] (u1) [label=above:$2$] {$u$};
\node[fixed] (u2) [right=1cm of u1, label=above:$2$] {$u'$};
\node[normal] (v1) [below=1cm of u1, label=left:$c$] {$v_1$};
\node[normal] (v2) [below=1cm of u2, label=right:$\lnot c$] {$v_2$};
\node[normal] (v3) [below left=1cm of v1, label=above:$2$] {$v_3$};
\node[] (v5) [below left=1cm of v3] {$x$};
\node[fixed] (u3) [below right=1cm of v3, label=above:$0$] {$u''$};
\node[normal] (v4) [below right=1cm of v2, label=above:$2$] {$v_4$};
\node[] (v6) [below left=1cm of v4] {$y$};
\node[fixed] (u4) [below right=1cm of v4, label=above:$0$] {$u'''$};

\draw (u1) -- (v1);
\draw (u2) -- (v2);
\draw (v1) -- (v2);
\draw (v1) -- (v3);
\draw (v3) -- (v5);
\draw (v3) -- (u3);
\draw (v4) -- (v2);
\draw (v4) -- (v6);
\draw (v4) -- (u4);

\end{tikzpicture}
}
    \caption{Labels $\pi(v)$ are indicated above the vertices.}
    \label{fig_SU3P_d3_2cls}
\end{figure}
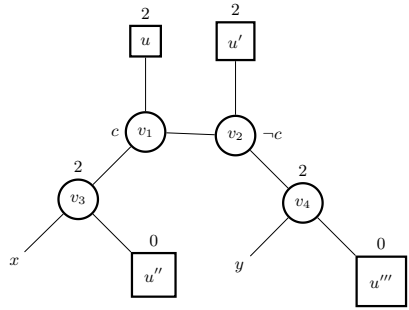

\subsection{Omitted Proofs of SU2P} \label{sec_app_prf2}

\begin{lemma}[Impagliazzo, Paturi and Zane \cite{10.5555/795664.796440}]
    \label{lem_paturi}
    There exists a constant $d \in \mathbb{N}^+$ such that $c_{3\text{-SAT}} = c_{(d, 3)\text{-SAT}}$.
\end{lemma}

\begin{proof}[Proof of Lemma~\ref{lem_33}]
    Let $d$ be the constant such that $c_{3\text{-SAT}} = c_{(d,3)\text{-SAT}}$ from Lemma~\ref{lem_paturi}. Given a $d$-occur-$3$-SAT instance $\phi$, for each variable $x_i$ appearing $j_i \leq d$ times, we introduce $j_i$ new variables $y_{i,1}, \dots, y_{i,j_i}$ to replace literals and add the implication cycle $(\neg x_i \lor y_{i,1}) \land (\neg y_{i,1} \lor y_{i,2}) \land \dots \land (\neg y_{i,j_i} \lor x_i)$.
    
    This ensures $x_i = y_{i,1} = \dots = y_{i,j_i}$ and that each variable occurs at most thrice in $\phi'$. The variable count $n'$ of $\phi'$ satisfies $n' = \sum_{i=1}^n (j_i + 1) \leq (d+1)n$. Under ETH, solving $\phi$ requires $\Omega(c_{(d,3)\text{-SAT}}^n) = \Omega(c_{3\text{-SAT}}^n)$ time. If $(3,3)\text{-SAT}$ is solvable in $O(c_{(3,3)\text{-SAT}}^{n'})$ time, then:
    $$c_{3\text{-SAT}}^n \leq c_{(3,3)\text{-SAT}}^{(d+1)n} \implies c_{3\text{-SAT}}^{1/(d+1)} \leq c_{(3,3)\text{-SAT}} $$
\end{proof}

\begin{proof}[Proof of Observation~\ref{ob_SU2P_d2}]
    If $d(u)=2$ and $\pi(u)=a$, the condition $d_a(u) < d_b(u)$ implies $d_a(u)=0$ and $d_b(u)=2$. Thus, the neighbourhood of $u$ is uniquely determined once $\pi(u)$ is fixed.
\end{proof}

\begin{proof}[Proof of Observation~\ref{ob_SU2P_d3}]
    If $d(u)=3$ and $\pi(u)=a$, the condition $d_a(u) < d_b(u)$ implies $d_a(u) \in \{0, 1\}$. This allows for two possible local configurations: either all three neighbours have the opposite label, or two neighbours have the opposite label and one has the same label.
\end{proof}

\begin{proof}[Proof of Proposition ~\ref{prop_SU2P_d3_cst}]
    ($\Rightarrow$) Suppose the constant gadget admits a legal strictly unfriendly 2-partition $\pi$. Consider a degree-2 vertex $u_{2j}$. By Observation~\ref{ob_SU2P_d2}, its neighbours $u_{2j-1}$ and $u_{2j+1}$ must be assigned a label different from $\pi(u_{2j})$. This local constraint propagates through the even cycle, forcing all odd-indexed vertices to share the same label $a \in \{0,1\}$, and all even-indexed vertices to share the label $b = 1-a$. For an odd-indexed vertex $u_{2j+1}$, let its cycle neighbours be $u_{2j}$ and $u_{2j+2}$. Since $\pi(u_{2j}) = \pi(u_{2j+2}) = b$ and $\pi(u_{2j+1}) = a$, the vertex $u_{2j+1}$ already has two neighbours having the opposite label. According to Observation~\ref{ob_SU2P_d3}, a degree-3 vertex is legal if it has at least two neighbours having the opposite label. Therefore, the label of its third (external) neighbour does not affect the legality of $u_{2j+1}$. This ensures that the constant gadget provides stable supply of constant labels for variable and clause gadgets.

    $(\Leftarrow)$ Conversely, the legality of all vertices under the alternating label assignment follows directly from Observations~\ref{ob_SU2P_d2} and ~\ref{ob_SU2P_d3}.
\end{proof}

\begin{proof}[Proof of Proposition ~\ref{prop_SU2P_d3_var}]
    ($\Rightarrow$) Suppose the gadget admits a legal strictly unfriendly 2-partition $\pi$. Consider the degree-2 auxiliary vertex $v_4$. By Observation~\ref{ob_SU2P_d2}, its neighbours $v_1$ and $v_2$ must have the opposite label of $v_4$. This immediately forces $\pi(v_1) = \pi(v_2)$, ensuring that both occurrences of the literal $x$ receive the same truth value.
    Now consider the degree-3 vertex $v_3$. Let $\pi(v_1) = \pi(v_2) = a$. If $v_3$ were assigned the same label $a$, it would have at least two neighbours ($v_1$ and $v_2$) having its own label, violating Observation~\ref{ob_SU2P_d3}. Thus, we must have $\pi(v_3) = 1-a$, satisfying the polarity requirement between $x$ and $\neg x$.

    $(\Leftarrow)$ Conversely, the legality of the gadget under the specified labels follows immediately from Observations~\ref{ob_SU2P_d2} and ~\ref{ob_SU2P_d3}.
\end{proof}

\begin{proof} [Proof of Proposition ~\ref{prop_SU2P_d3_3cls}]
    $(\Rightarrow)$ Suppose the gadget admits a legal strictly unfriendly 2-partition $\pi$. Recall that $\pi(u) = \pi(u') = 1$. Iteratively applying Observation~\ref{ob_SU2P_d2} on degree-2 vertices from $v_3$ to $v_2$, we immediately obtain the forced labels: $\pi(v_3) = 0, \pi(v_2) = 1$, and consequently $\pi(v_1) = 0$. For $v_1$ to be legal with $\pi(v_1)=0$, Observation~\ref{ob_SU2P_d3} requires that it has at most one neighbour with label $0$. Since its neighbour $v_2$ is already labeled $1$, at least one of its remaining neighbours $z$ or $v_4$ must be labeled $1$. We consider two cases for the literal $z$:
    \begin{itemize}
        \item If $\pi(z) = 1$, the clause is satisfied regardless of $\pi(x)$ and $\pi(y)$.
        \item If $\pi(z) = 0$, then we must have $\pi(v_4) = 1$. By Observation~\ref{ob_SU2P_d3}, for the degree-3 vertex $v_4$ (adjacent to $v_1, v_5, u$), since $\pi(v_1)=0$ and $\pi(u)=1$, $v_4$ is legal only if $\pi(v_5) = 0$. For $v_5$ (with label $0$) to be legal, it must have at most one neighbour with label $0$. This forces at least one of its literal neighbours $x$ or $y$ to be $1$.
    \end{itemize}
    In both cases, at least one of $\{x, y, z\}$ is labeled $1$, satisfying the clause.
    
    $(\Leftarrow)$ Conversely, if the clause $(x \lor y \lor z)$ is satisfied, it is straightforward to find a label assignment for the internal vertices $\{v_1, \dots, v_5\}$ that satisfies Observations~\ref{ob_SU2P_d2} and~\ref{ob_SU2P_d3}, thereby extending the truth assignment to a legal partition of the gadget.
\end{proof}

\subsection{Omitted Proofs of SU3P} \label{sec_app_prf3}

\begin{proof}[Proof of Observation~\ref{ob_SU3P_zero}]
    If $d_{\pi(v)}(v) \geq 1$, the strictly unfriendly condition requires $d_j(v) \geq 2$ for each $j \neq \pi(v)$, implying the contradiction $d(v) = \sum_{i=0}^2 d_i(v) \geq 1 + 2 + 2 = 5 > \Delta$.
\end{proof}

\begin{proof}[Proof of Observation~\ref{ob_SU3P_d2}]
    If $d(u)=2$ and $\pi(u)=2$, then the neighbours $v, w$ must satisfy $\pi(v)=0$ and $\pi(w)=1$ (or vice versa) to ensure $d_2(u)=0 < d_0(u)=1$ and $d_2(u)=0 < d_1(u)=1$.
\end{proof}

\begin{proof}[Proof of Observation~\ref{ob_SU3P_d3}]
    For a vertex $u$ with $\pi(u)=2$ and neighbours $\{v, w, p\}$, the neighbour labels must form the multiset $\{0, 0, 1\}$ or $\{0, 1, 1\}$ up to permutation.
\end{proof}

\begin{proof}[Proof of Observation~\ref{ob_SU3P_ppg}]
    Suppose $\pi(v_1)=2$. By Observation~\ref{ob_SU3P_d2}, the neighbours of $w_1$ (which are $v_1$ and $w_2$) must have distinct labels from $\pi(w_1)$. If we set $\pi(w_1)=0$, then $\pi(w_2)=1$. Applying the same observation to $w_2$, its neighbours $w_1$ and $v_2$ must have distinct labels from $\pi(w_2)=1$. Since $\pi(w_1)=0$, we are forced to have $\pi(v_2)=2$.
\end{proof}

\begin{proof} [Proof of Proposition ~\ref{prop_SU3P_d3_cst}]
    Consider any two consecutive anchor vertices $u_{3j+1}$ and $u_{3j+4}$. They are connected by a path of two degree-2 vertices, $(u_{3j+2}, u_{3j+3})$. By Observation~\ref{ob_SU3P_ppg}, any two vertices connected by such a path must be assigned the same label in a legal 3-partition. By induction, all anchor vertices in the cycle must share a single label $L \in \{0, 1, 2\}$. Furthermore, according to Observation~\ref{ob_SU3P_d3}, the legality of a degree-3 anchor vertex is maintained as long as its cycle neighbours provide two distinct labels and the external neighbour vertex has a label different from the anchor vertex.
\end{proof}

\begin{proof} [Proof of Proposition ~\ref{prop_SU3P_d3_var}]
    $(\Rightarrow)$ Suppose the gadget admits a legal strictly unfriendly 3-partition $\pi$. First, consider the literal synchronization. Vertices $v_1$ and $v_4$ are connected by a path of two degree-2 vertices $(v_2, v_3)$. By Observation~\ref{ob_SU3P_ppg}, this forces $\pi(v_1) = \pi(v_4)$. Similarly, the path $(v_6, v_7)$ forces $\pi(v_5) = \pi(v_8)$. Next, we establish the Boolean domain restriction. Vertex $v_4$ is a degree-3 vertex adjacent to the constant-2 anchor $u$. If $\pi(v_4) = 2$, then $v_4$ would violate Observation~\ref{ob_SU3P_zero}. Thus, $\pi(v_4) \in \{0, 1\}$. By the same logic applied to $v_5$ and $u'$, we have $\pi(v_5) \in \{0, 1\}$. Finally, since $v_4$ and $v_5$ are adjacent, they must be assigned different labels (otherwise violate Observation~\ref{ob_SU3P_zero}). Given that both are restricted to $\{0, 1\}$, it follows that $\{\pi(v_4), \pi(v_5)\} = \{0, 1\}$, ensuring that $x$ and $\neg x$ always receive opposite truth values.

    $(\Leftarrow)$ Conversely, if the truth assignment is consistent across literal occurrences (i.e., $\pi(v_1) = \pi(v_4) \in \{0, 1\}$ and $\pi(v_5) = \pi(v_8) = 1 - \pi(v_4)$), it is straightforward to extend the labeling to a legal partition.
\end{proof}

\begin{proof} [Proof of Proposition ~\ref{prop_SU3P_d3_3cls}]
    $(\Rightarrow)$ Suppose the gadget admits a legal strictly unfriendly 3-partition $\pi$. First, consider the adjacent pair $v_1$ and $v_2$. Similarly as $v_4,v_5$ in the variable gadget, we have $\{\pi(v_1),\pi(v_2)\}=\{0,1\}$. Let $\pi(v_1)=c$ and $\pi(v_2)=\lnot c$. Next, consider the internal vertices $v_3$ and $v_4$. Since all literal vertices $x, y, z$ are restricted to $\{0, 1\}$ by the variable gadgets, $v_3$ is adjacent to three vertices $\{v_1, x, y\}$ all within the set $\{0, 1\}$. To satisfy Observation~\ref{ob_SU3P_zero}, $v_3$ must be assigned the third label, $\pi(v_3) = 2$. By identical logic, $v_4$ must also be assigned $\pi(v_4) = 2$. We now analyze the satisfiability based on the value of $c$:
    \begin{itemize}
        \item $[c=1]:$ Then $\pi(v_2) = 0$. For $v_4$ (labeled 2) to be legal, it must have at least one of its neighbours $v_2, z, u''$ with label $1$. Since $\pi(v_2) = \pi(u'') = 0$, the vertex $z$ must be labeled $1$.
        \item $[c=0]:$ Then $\pi(v_1) = 0$. For $v_3$ (labeled 2) to be legal, its neighbours $\{v_1, x, y\}$ must include at least one neighbour with label $1$. Since $\pi(v_1) = 0$, at least one of $\{x, y\}$ must be labeled $1$.
    \end{itemize}
    In either case, the assignment $1$ is present among the literals, satisfying the clause.

    $(\Leftarrow)$ Conversely, if the clause $(x \lor y \lor z)$ is satisfied, a valid label assignment for the internal vertices $\{v_1, \dots, v_4\}$ in $\{0, 1, 2\}$ can be readily verified to exist by applying the local constraints from Observation~\ref{ob_SU3P_d3}, thereby completing the legal partition.
\end{proof}

\subsection{SU2P without degree restriction} \label{sec_app_SU2P}

\begin{theorem} \label{thm_SU2P_ud}
    SU2P is NP-hard on general graphs~\cite{kun2013anti}. Furthermore, assuming ETH, it cannot be solved in $2^{o(n)}$ time; specifically, $c_{(3,3)\text{-SAT}}^{2/7} \leq c_{\text{SU2P}}$.
\end{theorem}

We introduce gadgets before we go into the full reduction. 

\medskip
\noindent \textbf{Constant Gadget.} The construction (shown in Figure~\ref{fig_SU2P_ud_cst}) needs five normal vertices $u_1,u_2,u_3,u_4,u_5$. Vertices $u_4,u_5$ only have neighbours in this gadget, while $u_1,u_2,u_3$ can have outside neighbours. As there is no degree restriction, $u_1,u_2,u_3$ can be reused multiples times as long as their labeling remains legal. Thus, the construction of constant vertices contains only five vertices in total.

We claim that, if the constant gadget is legally partitioned, $u_4,u_5$ have the same label, $u_1,u_2,u_3$ have the same label, and $u_1,u_4$ have different labels. Note that, $u_4,u_5$ do not have other neighbours; $u_1,u_2,u_3$ each has multiple neighbours: $n$ variable vertices (labeled $0$ or $1$, see proposition~\ref{prop_SU2P_ud_var}) and $m$ clause vertices (all labeled $0$, see proposition~\ref{prop_SU2P_ud_3cls}). Thus, as long as $n\leq m$ (which can always be achieved by introducing dummy variables and clauses), labeling for $u_1,u_2,u_3$ is legal.
\begin{figure}[H]
    \centering
    \resizebox{0.37\textwidth}{!}{
    \begin{tikzpicture}[
normal/.style = {draw, regular polygon,regular polygon sides=360, very thick},
fixed/.style={draw, regular polygon,regular polygon sides=4, very thick},
balance/.style = {draw, regular polygon,regular polygon sides=360, fill=black},
]
\node[fixed] (u1) [label=below:$0$] {$u_4$};
\node[fixed] (u3) [above=1cm of u1, label=above:$1$] {$u_1$};
\node[fixed] (u4) [right=1cm of u3, label=above right:$1$] {$u_2$};
\node[fixed] (u5) [right=1cm of u4, label=above:$1$] {$u_3$};
\node[fixed] (u2) [below=1cm of u5, label=below:$0$] {$u_5$};
\node[] (e3) [left=1cm of u3] {};
\node[] (e4) [above=1cm of u4] {};
\node[] (e5) [right=1cm of u5] {};

\draw (u1) -- (u2);
\draw (u1) -- (u3);
\draw (u1) -- (u4);
\draw (u1) -- (u5);
\draw (u2) -- (u3);
\draw (u2) -- (u4);
\draw (u2) -- (u5);
\draw (u3) -- (e3);
\draw (u4) -- (e4);
\draw (u5) -- (e5);

\end{tikzpicture}
}
    \caption{}
    \label{fig_SU2P_ud_cst}
\end{figure}

\begin{proposition} \label{prop_SU2P_ud_cst}
    $\pi(u_1)=\pi(u_2)=\pi(u_3)=1-\pi(u_4)=1-\pi(u_5)$ iff the constant gadget is legally partitioned.
\end{proposition}
\begin{proof}
    $(\Leftarrow)$ Without loss of generality, we assume $\pi(u_4)=0$. To make $u_4$ legal, there are two cases.
    \begin{itemize}
        \item \textbf{All $u_1,u_2,u_3,u_5$ are labeled $1$:} Then $u_5$ has three neighbours with label $1$ and thus is illegal.
        \item \textbf{One of $u_1,u_2,u_3,u_5$ is labeled $0$:} This case has two subcases.
        \begin{itemize}
            \item \textbf{$\pi(u_5)=0$:} Then both $u_4$ and $u_5$ have three neighbours $u_1,u_2,u_3$ with label $1$ and one neighbour with label $0$, and thus are legal.
            \item \textbf{One of $u_1,u_2,u_3$ is labeled $0$:} Then $u_5$ has two neighbours with label $0$ and two neighbour with label $1$s, and thus is illegal.
        \end{itemize}
    \end{itemize}
    The only legal case is when $\pi(u_1)=\pi(u_2)=\pi(u_3)=1-\pi(u_4)=1-\pi(u_5)$. As we claim $\pi(u_4)=0$, $u_1,u_2,u_3$ can be used as constant $1$ vertices.

    $(\Rightarrow)$ It is straightforward that when $\pi(u_1)=\pi(u_2)=\pi(u_3)=1-\pi(u_4)=1-\pi(u_5)$, the constant gadget is legally partitioned.
    
    This completes the proof.
\end{proof}

\medskip
\noindent \textbf{Variable Gadget.} For a variable $x$, the construction (shown in Figure~\ref{fig_SU2P_ud_var}) needs three constant vertices $u,u',u''$ with $\pi(u)=\pi(u')=\pi(u'')=1$, and two normal vertices $v_1,v_2$ representing $x,\lnot x$ respectively. Note that, if $x$ occurs once, then $u'$ is not there. Vertex $v_1$ is connected with the same number of constant $1$ vertices and clause vertices where $x$ appears. While $v_2$ is connected with the same number of constant $1$ vertices and clause vertices where $\lnot x$ appears. Vertices $u,u',u''$ are counted when calculating constant vertices. Thus, each variable needs $|\{v_1,v_2\}|=2$ vertices.

We claim that, if the variable gadget is legally partitioned, $\pi(v_1)=1-\pi(v_2)$, and thus $v_1,v_2$ can represent $x,\lnot x$ respectively.
\begin{figure}[H]
    \centering
    \resizebox{0.38\textwidth}{!}{
    \begin{tikzpicture}[
normal/.style = {draw, regular polygon,regular polygon sides=360, very thick},
fixed/.style={draw, regular polygon,regular polygon sides=4, very thick, minimum size=1.2cm},
balance/.style = {draw, regular polygon,regular polygon sides=360, fill=black},
]
\node[normal] (v1) [label=left:$x$] {$v_1$};
\node[] (u4) [below=1cm of v1] {to clause};
\node[] (u5) [right=1cm of u4] {to clause};
\node[] (u6) [right=1cm of u5] {to clause};
\node[normal] (v2) [above=1cm of u6, label=right:$\lnot x$] {$v_2$};
\node[fixed] (u1) [above=1cm of v1, label=above:$1$] {$u$};
\node[fixed] (u2) [right=1cm of u1, label=above:$1$] {$u'$};
\node[fixed] (u3) [above=1cm of v2, label=above:$1$] {$u''$};

\draw (v2) -- (v1);
\draw (u1) -- (v1);
\draw (u2) -- (v1);
\draw (u3) -- (v2);
\draw (u4) -- (v1);
\draw (u5) -- (v1);
\draw (u6) -- (v2);

\end{tikzpicture}
}
    \caption{}
    \label{fig_SU2P_ud_var}
\end{figure}

\begin{proposition} \label{prop_SU2P_ud_var}
    $\pi(v_1)=1-\pi(v_2)$ iff the variable gadget is legally partitioned.
\end{proposition}
\begin{proof}
    $(\Leftarrow)$ The clause gadget ensures that vertices connected to $v_1$ and $v_2$ are guaranteed to be labeled $0$. This will be proven when introducing the construction for clauses. For now, without loss of generality, we assume $x$ appears in two clauses and $\lnot x$ appears in one clause. Thus, $v_1$ two neighbours with label $1$ and two neighbour with label $0$s. To make $v_1$ legal, $v_2$ must be labeled different from $v_1$. Similarly, $v_1$ must have different label with $v_2$ to make $v_2$ legal.

    $(\Rightarrow)$ It is straightforward that when $\pi(v_1)=1-\pi(v_2)$, the variable gadget is legally partitioned.

    This completes the proof.
\end{proof}

\medskip
\noindent \textbf{Clause Gadget.} For a clause $(x\lor y\lor z)$, the construction (shown in Figure~\ref{fig_SU2P_ud_3cls}) needs three constant vertices $u,u',u''$ with $\pi(u)=\pi(u')=\pi(u'')=1$ and one normal vertex $v_1$, which is connected with $x,y,z$. Vertices $u,u',u''$ are counted when calculating constant vertices. Thus, each three-literal clause needs $|\{v_1\}|=1$ vertex.

We claim that the clause gadget is legally partitioned iff at least one of $x,y,z$ should be $1$ and thus the clause is satisfied.
\begin{figure}[H]
    \centering
    \resizebox{0.27\textwidth}{!}{
    \begin{tikzpicture}[
normal/.style = {draw, regular polygon,regular polygon sides=360, very thick},
fixed/.style={draw, regular polygon,regular polygon sides=4, very thick, minimum size=1.2cm},
balance/.style = {draw, regular polygon,regular polygon sides=360, fill=black},
]
\node[normal] (v1) [label=left:$0$] {$v_1$};
\node[fixed] (u2) [above=1cm of v1, label=above:$1$] {$u'$};
\node[fixed] (u1) [left=1cm of u2, label=above:$1$] {$u$};
\node[fixed] (u3) [right=1cm of u2, label=above:$1$] {$u''$};
\node[] (v3) [below=1cm of v1] {$y$};
\node[] (v2) [left=1cm of v3] {$x$};
\node[] (v4) [right=1cm of v3] {$z$};

\draw (v2) -- (v1);
\draw (v3) -- (v1);
\draw (v4) -- (v1);
\draw (u2) -- (v1);
\draw (u1) -- (v1);
\draw (u3) -- (v1);

\end{tikzpicture}
}
    \caption{}
    \label{fig_SU2P_ud_3cls}
\end{figure}

\begin{proposition} \label{prop_SU2P_ud_3cls}
    The clause $(x\lor y\lor z)$ is satisfied iff the clause gadget is legally partitioned.
\end{proposition}
\begin{proof}

    $(\Leftarrow)v_1$ has degree six and three neighbours with label $1$s. Then if $\pi(v_1)=1$, whatever $x,y,z$ are, $v_1$ is illegal. Thus, $v_1$ must have label $0$. In the case where $x=y=z=0$, $v_1$ has three neighbours with label $1$ and three neighbours with label $0$s, and thus is illegal. Therefore, at least one of $x,y,z$ should be $1$ to make $v_1$ legal.

    $(\Rightarrow)$ If $(x\lor y\lor z)$ is satisfied, then the legal partition of the clause gadget can be obtained accordingly.
\end{proof}
The gadget for a 2-literal clause $(x\lor y)$ is derived from the 3-clause gadget by deleting vertices $u''$ and $z$. Both 2-clause and 3-clause gadgets need 1 vertex ($v_1$) only.

After all gadgets are introduced, we show the full reduction.
\begin{proof}[Proof of Theorem~\ref{thm_SU2P_ud}]
    Given an instance $\phi$ of $(3,3)$-SAT with $n$ variables and $m$ clauses, we construct the graph $G_\phi$ by assembling variable gadgets, clause gadgets, and the necessary constant neighbours as described earlier in the section. By Propositions~\ref{prop_SU2P_ud_cst}, ~\ref{prop_SU2P_ud_var}, and ~\ref{prop_SU2P_ud_3cls}, it follows that $G_\phi$ admits a strictly unfriendly 2-partition if and only if $\phi$ is satisfiable, giving SU2P is NP-hard.

    We now determine the total number of vertices $N = |V(G_\phi)|$ relative to the number of variables $n$ in $\phi$.
    \begin{itemize}
        \item Variable Gadgets: For each of the $n$ variables, the gadget introduces exactly $2$ vertices. Total: $2n$.
        \item Clause Gadgets: For each of the $m$ clauses, the gadget introduces exactly $1$ vertices. Total: $m\leq1.5n$ (by Lemma~\ref{lem_nm}).
        \item Constant Gadget: There are 5 vertices for the constant gadget.
    \end{itemize}
    Summing these, we obtain $N \leq 3.5n + 5$. The reduction is clearly linear in $n$. Assuming ETH, $(3,3)$-SAT cannot be solved in $2^{o(n)}$ time. Substituting $n \geq 2N/7$ (5 constant vertices are ignored, as they don't effect the limiting ratios), we conclude that SU2P cannot be solved in $2^{o(N)}$ time on general graphs, specifically:
    $$ c_{(3,3)\text{-SAT}}^n \leq c_{\text{SU2P}}^{3.5n} \implies c_{\text{SU2P}} \geq c_{(3,3)\text{-SAT}}^{2/7}.$$
    This completes the proof.
\end{proof}

\subsection{SUkP ($k\geq3$) without degree restriction} \label{sec_app_SUkP}

\begin{theorem} \label{thm_sukp}
    SUkP ($k\geq3$) is NP-hard on general graphs~\cite{kun2013anti}. Furthermore, assuming ETH, it cannot be solved in $2^{o(n)}$ time; specifically, $c_{(3,3)\text{-SAT}}^{2/7} \leq c_{\text{SUkP}}$ for all $k\geq3$.
\end{theorem}

We introduce gadgets before we go into the full reduction. 

\medskip
\noindent \textbf{Constant Gadget.} Such construction is shown in Figure~\ref{fig_SUkP_cst}. It is a $k$-complete graph consisting of $u_1,u_2,...,u_k$. All vertices except $u_2$ can have extra neighbours. In total, $k$ vertices are needed to construct such gadget.
\begin{figure}[H]
    \centering
    \resizebox{0.43\textwidth}{!}{
    \begin{tikzpicture}[
normal/.style = {draw, regular polygon,regular polygon sides=360, very thick},
fixed/.style={draw, regular polygon,regular polygon sides=4, very thick},
balance/.style = {draw, regular polygon,regular polygon sides=360, fill=black},
omit/.style={draw=white, regular polygon,regular polygon sides=4, very thick},
]
\node[omit] (u) [] {$......$};
\node[fixed] (uk) [above left=3cm of u, label=above:$k-1$] {$u_k$};
\node[fixed] (u3) [above right=3cm of u, label=above:$2$] {$u_3$};
\node[fixed] (u1) [above right=1cm of uk, label=above:$0$] {$u_1$};
\node[fixed] (u2) [above left=1cm of u3, label=above:$1$] {$u_2$};
\node[] (e1) [above left=1cm of u1] {};
\node[] (e3) [right=1cm of u3] {};
\node[] (ek) [left=1cm of uk] {};

\draw (u1) -- (u2);
\draw (u1) -- (u3);
\draw (u1) -- (u);
\draw (u1) -- (uk);
\draw (u2) -- (u3);
\draw (u2) -- (u);
\draw (u2) -- (uk);
\draw (u3) -- (u);
\draw (u3) -- (uk);
\draw (u) -- (uk);
\draw (u1) -- (e1);
\draw (u3) -- (e3);
\draw (uk) -- (ek);

\end{tikzpicture}
}
    \caption{Labels $\pi(v)$ are indicated above the vertices.}
    \label{fig_SUkP_cst}
\end{figure}
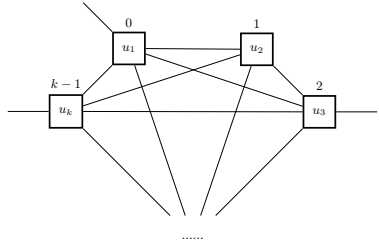

\begin{proposition} \label{prop_SUkP_cst}
    $\pi(u_1)\neq \pi(u_2)\neq...\neq \pi(u_k)$ iff this gadget is legally partitioned.
\end{proposition}
Follows as labels of $u_2$ and its neighbours need to be different. Without loss of generality, we assume $\pi(u_i)=i-1$, so that we have constant vertices with labels $0,2,...,k-1$.

\medskip
\noindent \textbf{Variable Gadget.} For every variable $x$, the construction needs $k-2$ constant vertices with labels $2,3,...,k-1$, and two normal vertices $v_1,v_2$ representing $x,\lnot x$. All anchor vertices have an edge to $v_1$ and $v_2$. The gadget is shown in Figure~\ref{fig_SUkP_var} assuming $x$ occurs twice in $\phi$. All the constant vertices are counted previously, so each variable needs two vertices in total ($v_1,v_2$).
\begin{figure}[H]
    \centering
    \resizebox{0.63\textwidth}{!}{
    \begin{tikzpicture}[
normal/.style = {draw, regular polygon,regular polygon sides=360, very thick},
fixed/.style={draw, regular polygon,regular polygon sides=4, very thick, minimum size=2.4cm},
balance/.style = {draw, regular polygon,regular polygon sides=360, fill=black},
omit/.style={draw=white, regular polygon,regular polygon sides=4, very thick},
]
\node[fixed] (u1) [label=above:$2$] {$u_3$};
\node[normal] (v1) [above left=3cm of u1, label=left:$x$] {$v_1$};
\node[fixed] (u2) [right=2cm of u1, label=above:$3$] {$u_4$};
\node[normal] (v2) [above right=3cm of u2, label=right:$\lnot x$] {$v_2$};
\node[fixed] (uk-2) [above right=3cm of v1, label=above:$k-1$] {$u_{k}$};
\node[omit] (u3) [right=2cm of uk-2] {$......$};
\node[] (e2) [above left=1cm of v1] {to clause};
\node[] (e3) [below left=1cm of v1] {to clause};
\node[] (e4) [above right=1cm of v2] {to clause};

\draw (u1) -- (v1);
\draw (u1) -- (v2);
\draw (v2) -- (v1);
\draw (e2) -- (v1);
\draw (e3) -- (v1);
\draw (e4) -- (v2);
\draw (u2) -- (v1);
\draw (u2) -- (v2);
\draw (u3) -- (v1);
\draw (u3) -- (v2);
\draw (uk-2) -- (v1);
\draw (uk-2) -- (v2);

\end{tikzpicture}
}
    \caption{Labels $\pi(v)$ are indicated above the vertices.}
    \label{fig_SUkP_var}
\end{figure}

We claim that, if this gadget is legally partitioned, neither of $v_1,v_2$ can have labels $2,3,...,k-1$, and $v_1,v_2$ must have different labels, so that we can use $v_1,v_2$ to denote $x,\lnot x$ respectively.

\begin{proposition} \label{prop_SUkP_var}
    In any legal partition $\pi$ of the variable gadget, $\pi(v_1), \pi(v_2) \in \{0,1\}$. Then, $\pi(v_1) = 1 - \pi(v_2)$ if and only if the variable gadget is legally partitioned.
\end{proposition}
\begin{proof}
    $(\Leftarrow)$ All clause neighbours having label $2$ (see Proposition~\ref{prop_SUkP_cls}), and anchor vertices $u_3,u_4,\dots,u_{k}$ forces $\pi(v_1),\pi(v_2)\in\{0,1\}$, which guarantees the polarity of $x,\lnot x$. To satisfy $d_i(v_1)>d_{\pi(v_1)}(v_1)$ for all $2\leq i\leq k-1$, $\pi(v_2)\neq\pi(v_1)$. Thus $\{\pi(v_1),\pi(v_2)\}=\{0,1\}$.

    $(\Rightarrow)$ If $\pi(v_1)=1-\pi(v_2)$, it is straightforward that the gadget is legally partitioned.
\end{proof}

\medskip
\noindent \textbf{Clause Gadget.} For a clause $(x\lor y\lor z)$, the construction needs $k-2$ constant vertices $u_1,u_4,u_5,...,u_{k}$ with labels $0,3,4,...,k-1$, and one normal vertex $v_1$, which is connected with $x,y,z$. This gadget is shown in Figure~\ref{fig_SUkP_3cls}. All the constant vertices and variable vertices are counted previously, so the two-literal clause $(x\lor y\lor z)$ needs one vertex in total ($v_1$).
\begin{figure}[H]
    \centering
    \resizebox{0.33\textwidth}{!}{
    \begin{tikzpicture}[
normal/.style = {draw, regular polygon,regular polygon sides=360, very thick},
fixed/.style={draw, regular polygon,regular polygon sides=4, very thick, minimum size=2.4cm},
balance/.style = {draw, regular polygon,regular polygon sides=360, fill=black},
omit/.style={draw=white, regular polygon,regular polygon sides=4, very thick},
]
\node[normal] (v1) [label=above:$2$] {$v_1$};
\node[] (v2) [below left=2cm of v1] {$x$};
\node[] (v3) [below right=2cm of v1] {$z$};
\node[] (v4) [below=2cm of v1] {$y$};
\node[fixed] (u1) [left=2cm of v1, label=above:$0$] {$u_1$};
\node[fixed] (u2) [above left=2cm of v1, label=above:$3$] {$u_4$};
\node[omit] (u3) [above right=2cm of v1] {$......$};
\node[fixed] (uk-2) [right=2cm of v1, label=above:$k-1$] {$u_k$};

\draw (u1) -- (v1);
\draw (v2) -- (v1);
\draw (v3) -- (v1);
\draw (v4) -- (v1);
\draw (u2) -- (v1);
\draw (u3) -- (v1);
\draw (uk-2) -- (v1);

\end{tikzpicture}
}
    \caption{Labels $\pi(v)$ are indicated above the vertices.}
    \label{fig_SUkP_3cls}
\end{figure}

We claim that, this gadget is legally partitioned, iff at least one of $x,y,z$ is labeled $1$ and thus the clause is satisfied.

\begin{proposition} \label{prop_SUkP_cls}
    The clause $(x\lor y\lor z)$ is satisfied iff the clause gadget is legally partitioned.
\end{proposition}
\begin{proof}
    $(\Leftarrow)$ Consider $v_1$, as $\pi(x),\pi(y),\pi(z)$ are forced to be in $\{0,1\}$ by variable gadget and $d_i(v_1)=1$ for all $3\leq i\leq k-1$, we must have $\pi(v_1) = 2$ and $d_2(v_1)=0$. To satisfy $d_1(v_1)>d_2(v_1)$, at least one of $x, y, z$ is forced to have label 1.

    $(\Rightarrow)$ If the clause is satisfied, it comes straightforward the gadget is legally partitioned.
\end{proof}
The gadget for a 2-literal clause $(x\lor y)$ is derived from the 3-clause gadget by deleting vertex $z$. Both 2-clause and 3-clause gadgets need 1 vertex ($v_1$) only.

After all gadgets are introduced, we show the full reduction.
\begin{proof}[Proof of Theorem~\ref{thm_sukp}]
    Given an instance $\phi$ of $(3,3)$-SAT with $n$ variables and $m$ clauses, we construct the graph $G_\phi$ by assembling variable gadgets, clause gadgets, and the necessary constant neighbours as described earlier in the section. By Propositions~\ref{prop_SUkP_cst}, ~\ref{prop_SUkP_var}, and ~\ref{prop_SUkP_cls}, it follows that $G_\phi$ admits a strictly unfriendly k-partition if and only if $\phi$ is satisfiable, giving SUkP is NP-hard.

    We now determine the total number of vertices $N = |V(G_\phi)|$ relative to the number of variables $n$ in $\phi$.
    \begin{itemize}
        \item Variable Gadgets: For each of the $n$ variables, the gadget introduces exactly $2$ vertices. Total: $2n$.
        \item Clause Gadgets: For each of the $m$ clauses, the gadget introduces exactly $1$ vertex. Total: $m\leq1.5n$ (by Lemma~\ref{lem_nm}).
        \item Constant Gadget: There are $k$ vertices for the constant gadget.
    \end{itemize}
    Summing these, we obtain $N \leq 3.5n+k$. The reduction is clearly linear in $n$. Assuming ETH, $(3,3)$-SAT cannot be solved in $2^{o(n)}$ time. Substituting $n \geq 2N/7$ ($k$ constant vertices are ignored, as they don't effect the limiting ratios), we conclude that SU$k$P cannot be solved in $2^{o(N)}$ time on general graphs, specifically:
    $$ c_{(3,3)\text{-SAT}}^n \leq c_{\text{SUkP}}^{3.5n} \implies c_{\text{SUkP}} \geq c_{(3,3)\text{-SAT}}^{2/7}.$$
    This completes the proof.
\end{proof}

\end{document}